\documentclass{article}
\usepackage[utf8]{inputenc}
\usepackage{secdot, amsfonts, amsmath, amssymb, bm, breqn, mathtools, graphicx, derivative, caption, bbm, dsfont, float, subcaption, tikz, relsize, amsthm, fullpage, array, authblk, natbib}
\usepackage[title]{appendix}
\newtheorem{theorem}{Theorem}
\newtheorem{corollary}{Corollary}[theorem]
\newtheorem{lemma}{Lemma}
\makeatletter
\def\Vdots{\vbox{\baselineskip12\p@ \lineskiplimit\z@
  \kern1\p@\hbox{.}\hbox{.}\hbox{.}\hbox{.}\hbox{.}\hbox{.}\hbox{.}\hbox{.}\hbox{.}}}
\makeatother

\allowdisplaybreaks 
\usepackage[colorlinks=false, linktocpage=true]{hyperref}
\hypersetup{
  colorlinks,
  citecolor=blue,
  linkcolor=blue,
  urlcolor=blue}

\usepackage[capitalise]{cleveref} 
\usetikzlibrary{automata, fit, shapes.geometric}
\crefrangelabelformat{equation}{(#3#1#4--#5#2#6)} 
\usepackage[export]{adjustbox} 
\usepackage[margin=1in]{geometry}
\newcommand{\beginsupplement}{%
        \setcounter{section}{0}
        \renewcommand{\thesection}{\Alph{section}} 
        \setcounter{lemma}{0}
        \setcounter{equation}{0}
        \setcounter{table}{0}
        \renewcommand{\thetable}{S\arabic{table}}%
        \setcounter{figure}{0}
        \renewcommand{\thefigure}{S\arabic{figure}}%
     }
 
\usepackage{setspace}
\begin{document}

\title{\bf Graphical Dirichlet Process for Clustering Non-Exchangeable Grouped Data}
\date{}

\author[1]{\small Arhit Chakrabarti}
\author[1, a]{\small Yang Ni}
\author[2, 3, b]{\small Ellen Ruth A. Morris}
\author[2, 3]{\small Michael L. Salinas}
\author[2, 3]{\small Robert S. Chapkin}
\author[1]{\small Bani K. Mallick}

\affil[1]{\footnotesize Department of Statistics, Texas A\&M University, College Station, TX}
\affil[2]{\footnotesize Department of Nutrition, Texas A\&M University, College Station, TX}
\affil[3]{\footnotesize Program in Integrative Nutrition \& Complex Diseases, Texas A\&M University, College Station, TX}
\affil[a]{\footnotesize All correspondence should be addressed to \href{mailto:yni@stat.tamu.edu}{yni@stat.tamu.edu}}
\affil[b]{\footnotesize Current address: Texas A\&M Veterinary Medical Diagnostic Laboratory, College Station, TX}

\maketitle

\begin{abstract}
\sloppy We consider the problem of clustering grouped data with possibly non-exchangeable groups whose dependencies can be characterized by a known directed acyclic graph. To allow the sharing of clusters among the non-exchangeable groups, we propose a Bayesian nonparametric approach, termed graphical Dirichlet process, that jointly models the dependent group-specific random measures by assuming each random measure to be distributed as a Dirichlet process whose concentration parameter and base probability measure depend on those of its parent groups. The resulting joint stochastic process respects the Markov property of the directed acyclic graph that links the groups. We characterize the graphical Dirichlet process using a novel hypergraph representation as well as the stick-breaking representation, the restaurant-type representation, and the representation as a limit of a finite mixture model. We develop an efficient posterior inference algorithm and illustrate our model with simulations and a real grouped single-cell dataset.
\end{abstract}
\sloppy \noindent%
\vfill

\newpage
\section{Introduction}
This article considers clustering of grouped data where the groups are \emph{non-exchangeable}. We are interested in settings where the data are \emph{partially exchangeable} \citep{definnetti}, which entails the exchangeability of the observations within each group but not across the groups. We consider dependent group–specific
random probability measures, thereby allowing the borrowing of information across non-exchangeable groups. We represent the dependencies among groups through a known \emph{directed acyclic graph} (DAG) with nodes denoting groups and directed edges denoting the group dependencies. Such data are abundant in many areas such as genomics. 
For example, our motivating application is a single-cell RNA-sequencing (scRNA-seq) study that aimed to investigate intestinal stem cell differentiation processes in mice with colorectal cancer. The experiments started from a baseline group where the mice were genetically wild-type, fed with a normal diet, and treated with no cancer therapy (placebo). Then to understand the main effects of genotype, diet, and cancer therapy on colonic crypt and tumor niche cell composition, the experimenters introduced three new groups of mice, each differing from the baseline group by exactly one factor (Apc knock-out, a high-fat diet, or a new cancer treatment AdipoRon). To determine the two-way interaction effects, three additional groups of mice were studied, each of which differed from the baseline group by two factors (e.g., mice with Apc knock-out, a high-fat diet, and no cancer treatment). Lastly, for a three-way interaction, they introduced the eighth group of mice with Apc knock-out, a high-fat diet, and the new treatment AdipoRon. The progression of these experiments from baseline to the study of main effects, two-way interactions, and three-way interactions manifests the non-exchangeability of the experimental groups (e.g., the baseline group is expected to be more similar to the ``main effect" groups than the ``three-way interaction" group). With this grouped scRNA-seq dataset, our goal is to cluster cells based on gene expression at the single-cell level within each experimental group while allowing information to be shared across these non-exchangeable groups with a novel DAG-based Bayesian nonparametric model.
\par
The \emph{Dirichlet process} (DP, \citealp{ferguson}) and its variations \citep{de2013gibbs,barrios2013modeling} have been the backbone of numerous model-based Bayesian nonparametric clustering methods \citep{hjort2010bayesian,muller2015bayesian}. The DP, $DP(\alpha_0, G_0)$, is a probability measure on probability measures, where $\alpha_0 > 0$ is the concentration parameter and $G_0$ is a base probability measure. There have been extensive studies on DP mixture models (\citealp{antonaik}; \citealp{escobar&west};  \citealp{maceacher_muller}), which enable clustering without having to fix the number of clusters \emph{a priori}. When there are groups present in the data, naively, one could consider either a separate DP mixture model for each group on one extreme or a single DP mixture model ignoring the groups on the other extreme. However, it is often desirable to identify group-specific clusters while allowing the groups to be linked so that clusters are comparable across groups. Given the goal of clustering the observations within each group, consider a set of random probability measures, $G_j$, one for each group $j$, where each $G_j$ is distributed as $DP(\alpha_{0j}, G_{0j})$ with group-specific concentration parameter $\alpha_{0j}$ and base probability measure $G_{0j}$. Many methods have been proposed to link these group-specific DPs to induce dependencies through the parameter $\alpha_{0j}$ and/or $G_{0j}$ (\citealp{cifarelli1978problemi}; \citealp{mallick}; \citealp{kleinmanibrahim}; 
\citealp{muller}). Perhaps one of the most well-known 
methods is the hierarchical Dirichlet process (HDP, \citealp{hdp}), which falls in the general framework of dependent DP \citep{DDP_MacEachern1, DDP_MacEachern2} and assumes each group-specific $G_j$ is distributed as $DP(\alpha_0, G_0)$ where $\alpha_0$ is the shared concentration parameter and $G_0$ is the shared base probability measure for all groups. They further assume that $G_0$ follows another DP, $G_0\sim DP(\gamma, H)$. Since draws from a DP are discrete with probability one \citep{Sethuraman}, the base measure $G_0$ is almost surely discrete, which ensures that the group-specific probability measure $G_j$ shares the same set of atoms. The corresponding HDP mixture model is thus capable of identifying group-specific clusters while borrowing strength across groups. By construction, HDP mixture model assumes that both the observations within each group and the groups are exchangeable. A similar approach with a different scope, the nested DP \citep{nestedDP}, assumes $G_j$ follows a DP-distributed random probability measure with another DP as the base measure, $G_j\sim Q$ and $Q\sim DP(\alpha_0,DP(\gamma,H))$. The nested structure allows for the clustering of groups but restricts the clusters of observations within each group to be either identical or completely unrelated across groups. Similarly to HDP, nested DP also assumes both the observations within each group and
the groups to be exchangeable. Several recent works \citep{latent_nDP,hidden_HDP} have been proposed to take advantage of the cluster-sharing feature of the HDP and the group-clustering feature of the nested DP. Dependent DP has also been extensively used to model random distributions with various other types of dependencies such as spatial and temporal dependencies 
\citep{anova_DDP,linear_DDP,dunson_herring, spatial_DDP,time_DP1,time_DP4, dahl_DDP};
see \citealp{DDP_review_paper} for a recent review of different dependent DPs.\par
In this paper, we are interested in modeling a set of group-specific random distributions of which the (conditional) dependencies can be characterized by a DAG whose nodes represent the groups. More precisely, we assume that the joint distribution of the set of group-specific 
 random distributions factorizes with respect to a DAG and, therefore, respects its Markov property (i.e., conditional independencies). We call such graph-dependent DP, the \emph{graphical Dirichlet process} (GDP). 
Using GDP as a mixing distribution, the GDP mixture model gives rise to group-specific clusters, which depend directly on their Markov blanket. The known flexibility of DAG in representing conditional dependencies renders the generality of the proposed GDP for modeling dependent random distributions and group-specific clusters beyond exchangeable groups. The use of DAGs in Bayesian nonparametrics has been considered in recent literature. \citealp{dey2022graphical} proposed a graphical Gaussian process to parsimoniously model multivariate spatial data by incorporating conditional independencies among variables encoded by a DAG. \citealp{gu2023bayesian} proposes a pyramid-shaped deep latent variable model for categorical data using a DAG to represent the layer-wise latent conditional dependency structure. These works showcased the usefulness of DAGs through their factorization in Bayesian nonparametric models. We also exploit such factorization in this paper but our model is significantly different from theirs in both approaches and scopes. For example, their graphs link variables whereas ours link groups, and they focus on the modeling of multivariate spatial fields or generative models for categorical data whereas we focus on clustering non-exchangeable grouped data. 
The proposed GDP is a general model. The well-known HDP is a special case of GDP with a specific type of DAG -- a fork, i.e., one parent node and many children nodes (detailed in Section \ref{sec:def}); see Figure \ref{fig::HDP introduction}. Several existing works on time-evolving topic models can also be reformulated using a DAG to capture the time-dependency structure \citep{timevaryingDP, dynamicHDP, evolutionaryHDP}. 
\par In this paper, we will characterize the proposed GDP by a novel \emph{hypergraph} representation. We will also provide several other representations analogous to those for the HDP, i.e., a stick-breaking representation, a restaurant-type representation, and a representation as an infinite limit of a finite mixture model. We develop efficient posterior sampling based on the SALTSampler \citep{SALTSampler} and a Blocked Gibbs sampler for DP/HDP \citep{ishwaranBGS, das_etal}. Simulations and the motivating grouped single-cell data are used to demonstrate our method. In summary, our main contribution is three-fold: 

\begin{enumerate}
    \item We propose a general Bayesian nonparametric approach, GDP, to incorporate non-exchangeable group dependencies for clustering. 
    \item We provide several characterizations of GDP, each providing a different perspective.
    \item We develop a Metropolis-within-blocked-Gibbs sampler for posterior inference. Since HDP is a special case of GDP, this also contributes to a new sampler for HDP. The difficulty of sampling the global weights for HDP is mitigated by using the specialized proposal of SALTSampler \citep{SALTSampler}.
\end{enumerate}

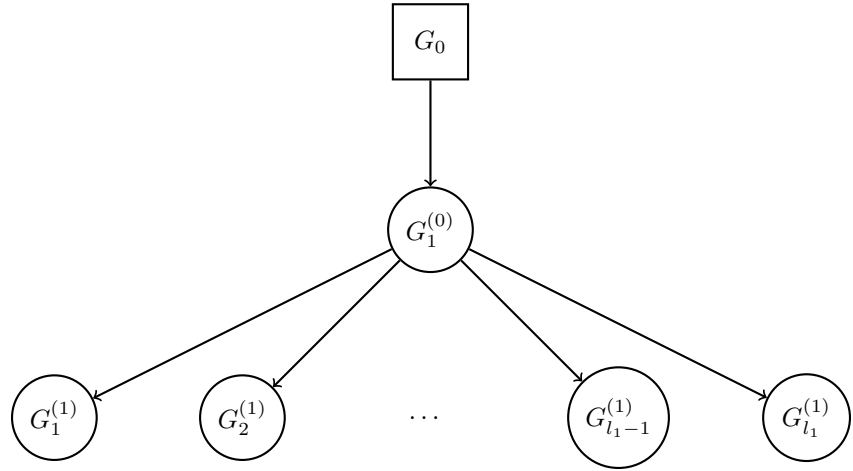
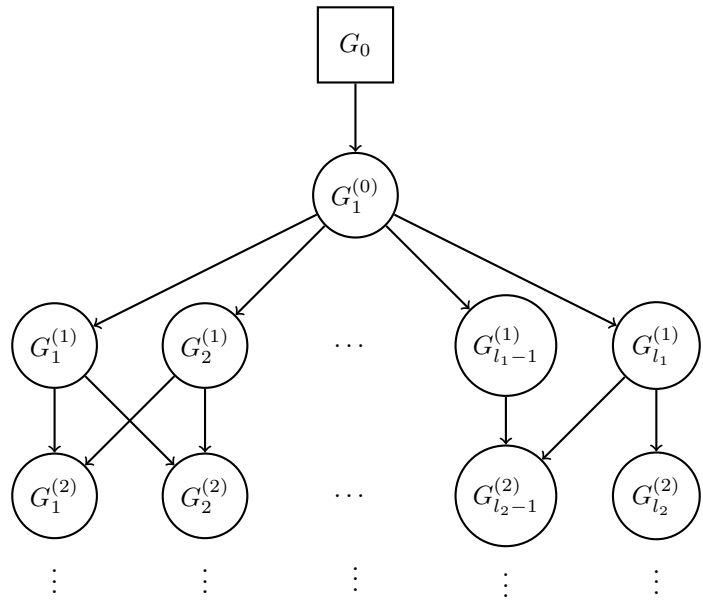
\begin{figure}[http]
\centering
    \begin{subfigure}[b]{0.8\textwidth}
        \centering
            \begin{tikzpicture}[node distance={25mm}, thick, main/.style = {draw, circle, minimum size=1cm, scale=1, transform shape}] 
\node[main, draw = none] (1) {}; 
\node[main, draw = none] (2) [right of = 1] {}; 
\node[main] (3) [right of=2] {$G_1^{(0)}$}; 
\node[main, rectangle] (3p) [above of=3] {$G_0$}; 
\node[main, draw = none] (3h1) [right of=3] {}; 
\node[main, draw = none] (3h2) [right of=3h1] {}; 
\node[main] (4) [below of=1] {$G_1^{(1)}$}; 
\node[main] (5) [right of=4] {$G_2^{(1)}$}; 
\node[main] (7) [below of=3h1] {$G_{l_1-1}^{(1)}$};
\node[main] (8) [below of=3h2] {$G_{l_1}^{(1)}$}; 
\draw[->] (3p) -- (3); 
\draw[->] (3) -- (4); 
\draw[->] (3) -- (5);
\draw[->] (3) -- (7);
\draw[->] (3) -- (8);
\path (5) -- node[auto=false]{\ldots} (7);
\end{tikzpicture} 
\caption{HDP}
\label{fig::HDP introduction}
    \end{subfigure}
    \par\bigskip
    \begin{subfigure}[b]{0.8\textwidth}
    \ContinuedFloat
    \centering
    \begin{tikzpicture}[node distance={20mm}, thick, main/.style = {draw, circle, minimum size=1cm, scale=1, transform shape}] 
\node[main, draw = none] (1) {}; 
\node[main, draw = none] (2) [right of = 1] {}; 
\node[main] (3) [right of=2] {$G_1^{(0)}$}; 
\node[main, rectangle] (3p) [above of=3] {$G_0$}; 
\node[main, draw = none] (3h1) [right of=3] {}; 
\node[main, draw = none] (3h2) [right of=3h1] {}; 
\node[main] (4) [below of=1] {$G_1^{(1)}$}; 
\node[main] (5) [right of=4] {$G_2^{(1)}$}; 
\node[main, draw = none] (6) [right of=5] {}; 
\node[main] (7) [below of=3h1] {$G_{l_1-1}^{(1)}$};
\node[main] (8) [below of=3h2] {$G_{l_1}^{(1)}$}; 
\node[main] (9) [below of=4] {$G_1^{(2)}$}; 
\node[main] (10) [right of=9] {$G_2^{(2)}$};
\node[main, draw = none] (11) [right of=10] {};
\node[main] (12) [below of=7] {$G_{l_2-1}^{(2)}$};
\node[main] (13) [below of=8] {$G_{l_2}^{(2)}$};

\draw[->] (3p) -- (3); 
\draw[->] (3) -- (4); 
\draw[->] (3) -- (5);
\draw[->] (3) -- (7);
\draw[->] (3) -- (8);
\draw[->] (4) -- (9);
\draw[->] (4) -- (10);
\draw[->] (5) -- (9);
\draw[->] (5) -- (10);
\draw[->] (7) -- (12);
\draw[->] (8) -- (13);
\draw[->] (8) -- (12);

\path (5) -- node[auto=false]{\ldots} (7);
\path (10) -- node[auto=false]{\ldots} (12);
\node[main, draw = none] (14) [below of=9] {}; 
\node[main, draw = none] (15) [below of=10] {}; 
\node[main, draw = none] (16) [below of=11] {}; 
\node[main, draw = none] (17) [below of=12] {}; 
\node[main, draw = none] (18) [below of=13] {}; 
\path (9) -- node[auto=false]{$\vdots$} (14);
\path (10) -- node[auto=false]{$\vdots$} (15);
\path (11) -- node[auto=false]{$\vdots$} (16);
\path (12) -- node[auto=false]{$\vdots$} (17);
\path (13) -- node[auto=false]{$\vdots$} (18);
\end{tikzpicture}
\caption{GDP for DAG with a unique root node.}
\label{fig::GDP introduction}
    \end{subfigure}
    \caption{Schematic illustration of HDP and GDP. HDP is a special case of GDP when the DAG is a fork. }
    \label{fig::HDP and GDP introduction}
\end{figure}

The remainder of the paper is organized as follows. Section \ref{sec:prelim} provides a brief overview of some preliminaries needed for the remainder of the paper. Section \ref{sec:gdp} introduces the proposed GDP and the corresponding nonparametric mixture model. We introduce the hyperpriors of our model and also present two lemmas, which are the backbone of our main result in Theorem \ref{main theorem}. In Section \ref{sec:representations}, we present different representations of the proposed GDP. In Section \ref{sec:simulations}, we provide simulations to illustrate our method. Section \ref{sec:read_data_analysis_main} presents a real data analysis using the proposed method on the motivating single-cell data. The paper concludes with a brief discussion in  Section \ref{sec:discussion}. 

\section{Preliminaries}\label{sec:prelim}
\subsection{Directed acyclic graph}
We first provide a brief background on DAG. 
Let $D = \left( V, E\right)$ be a DAG consisting of a set of nodes $V = \left\{1, 2, \dots , p\right\}$ and a set of directed edges $E \subset V \times V$ that does not contain any directed cycles. We denote a directed edge from the node $i$ to node $j$ by $j \leftarrow i$ and call $i$ a \emph{parent} of $j$. A node without parents is called a \emph{root}. For a DAG, there exists at least one root. Let $\bm{Y}=\{Y_1,\dots,Y_p\}$ be a set of random variables. 
 Every node $j\in V$ represents a random variable $Y_j$; later in this paper, $Y_j$ will be a random probability measure. In a DAG model, also known as a Bayesian network, the probability distribution $\mathcal{P}(\boldsymbol{Y})$ is assumed to factorize over $D$, $\mathcal{P}(\boldsymbol{Y}) = \prod_{j = 1}^{p}\mathcal{P}(Y_j \mid Y_{pa(j)})$, where $pa(j)=\{k\in V|j\gets k\}$ denotes the collection of parents of node $j$. This DAG factorization implies that the distribution $\mathcal{P}$ respects the conditional independence relationships encoded by the graph $D$ via the notion of d-separation \citep{pearl2009causality}; and vice versa. For instance, any node is conditional independent of its non-descendants given its parents, i.e., $Y_j\perp Y_{nd(j)}|Y_{pa(j)}$ for any $j\in V$ where $\perp$ denotes independence, $nd(j)=V\backslash de(j)\backslash \{j\}$ denotes the non-descendants of node $j$, and $de(j)=\{k\in V|k\gets\cdots\gets j\}$ denotes the descendants of node $j$. 
 A \emph{Markov blanket} of any node $j$ from $V$ is any subset $V_1$ of $V$ such that $Y_j\perp Y_{V \backslash V_1}|Y_{V_1}$. In other words, $V_1$ contains all the information in $V$ about the node $j$. 
 DAG models are convenient tools to parsimoniously specify a  multivariate distribution through its conditionals, which is especially useful in this paper for specifying a multivariate distribution of a set of random probability measures. \par

\subsection{Infinite mixture model}
Next, we present a brief overview of infinite mixture models for a single population, the DP mixture model, and for multiple exchangeable populations, the HDP mixture model. 

\subsubsection{Dirichlet process mixture model}
For a single population, let $x_i$ denote the $i$th realization of a random variable $X$. We consider a mixture model,
\begin{equation}
\label{DP mixture model}
\begin{aligned}
    \theta_i \mid G   & \overset{iid}{\sim} G,\\
    x_i \mid \theta_i & \overset{ind}{\sim} F(\theta_i),
\end{aligned}
\end{equation}
\noindent where $F(\theta_i)$ denotes the distribution of $x_i$ parameterized by $\theta_i$. The parameters $\theta_i$'s are conditionally independent given the prior distribution $G$. In a DP mixture model, $G$ is assigned a DP prior, $G \sim DP(\alpha_0, G_0)$ with concentration $\alpha_0$ and base probability measure $G_0$.\par
\citealp{Sethuraman} presented the \emph{stick-breaking representation} of the DP based on independent sequences of i.i.d. random variables $(\pi_k')_{k=1}^{\infty}$ and $(\phi_k)_{k=1}^{\infty}$,  which is given by,
\begin{align}
\label{DP stick breaking 1}
    \pi_k' & \mid \alpha_0\overset{iid}{\sim} Beta(1, \alpha_0), & & \phi_k\mid G_0 \overset{iid}{\sim} G_0,\\
    \label{DP stick breaking 2}
    \pi_k & = \pi_k'\prod_{l=1}^{k-1}(1-\pi_l'), & & G = \sum_{k=1}^{\infty}\pi_k\delta_{\phi_k},
\end{align}
where $\delta_{\phi}$ is a point mass at $\phi$ and $\phi_k$'s are called the \emph{atoms} of $G$. The sequence of random weights $\boldsymbol{\pi}=(\pi_k)_{k=1}^{\infty}$ constructed from \cref{DP stick breaking 1} and \cref{DP stick breaking 2} satisfies $\sum_{k=1}^{\infty}\pi_k = 1$ with probability one. The random probability measure on the set of integers is denoted by $\boldsymbol{\pi} \sim \text{GEM}(\alpha_0)$ for convenience where GEM stands for Griffiths, Engen and McCloskey \citep{pitman_GEM}. It is clear from Eq. \eqref{DP mixture model} and Eq. \eqref{DP stick breaking 2} that $\theta_i$ takes the value $\phi_k$ with probability $\pi_k$. Let $z_i$ be a categorical variable such that $z_i=k$ if $\theta_i=\phi_k$. An equivalent representation of a Dirichlet process mixture is given by, 
\begin{equation}
    \begin{aligned}
        \boldsymbol{\pi}\mid \alpha_0 & \sim \text{GEM}(\alpha_0), &  z_i \mid \boldsymbol{\pi} & \overset{iid}{\sim} \boldsymbol{\pi},\\
        \phi_k \mid G_0 & \overset{iid}{\sim} G_0,  & x_i\mid z_i, (\phi_k)_{k=1}^{\infty} & \overset{ind}{\sim} F(\phi_{z_i}).
    \end{aligned}
\end{equation}
\subsubsection{Hierarchical Dirichlet process mixture model}
Suppose observations are now organized into multiple exchangeable groups. Let $x_{ji}$ denote the observation $i$ from group $j$ and $\theta_{ji}$ denote the parameter specifying the mixture component associated with the corresponding observation. Let $F(\theta_{ji})$ denote the distribution of $x_{ji}$ given $\theta_{ji}$ and $G_j$ denote a prior distribution for $\theta_{ji}$. The group-specific mixture model is given by, 
\begin{equation}
\label{hdp mixture model}
    \begin{aligned}
        \theta_{ji}\mid G_j &\overset{ind}{\sim} G_j,\\
        x_{ji}\mid \theta_{ji} &\overset{ind}{\sim} F(\theta_{ji}).
    \end{aligned}
\end{equation}
As with the DP mixture model, when the random measures $G_j$'s are assigned an HDP prior,
\begin{equation}
\label{hdp eq1}
\begin{aligned}
    G_0 \mid \gamma, H & \sim DP(\gamma, H),\\
    G_j \mid \alpha_0, G_0 & \sim DP(\alpha_0, G_0),
\end{aligned}
\end{equation}
the corresponding mixture model is referred to as the HDP mixture model.
The global random probability measure $G_0$ is distributed as a DP with concentration parameter $\gamma$ and base probability measure $H$. The group-specific random measures $G_j$'s are conditionally independent given $G_0$ and hence are exchangeable \citep{definnetti}. They are distributed as DP with the base measure $G_0$ and some concentration parameter $\alpha_0$. The probability model \eqref{hdp mixture model} along with \eqref{hdp eq1} completes the specification of an HDP mixture model. Because DP-distributed $G_0$ is almost surely discrete, the atoms of $G_j$'s and hence the group-specific clusters are necessarily shared across groups. 
 
 \section{Graphical Dirichlet Process}\label{sec:gdp}
When groups are non-exchangeable (e.g., due to study design), the joint distribution of $G_j$'s specified by \eqref{hdp eq1} may not be appropriate.
 Our approach to the problem of sharing clusters among non-exchangeable groups is through specifying a general joint distribution of $G_j$'s that respect the Markov property of a DAG $D$ that links the groups.  
 We assume that the underlying DAG $D$ is known and we define the appropriate prior on the nodes of the DAG and refer to the resulting stochastic process on the graph as the graphical Dirichlet process (GDP). We show how this prior can be used in the non-exchangeable grouped mixture model setting. 

\subsection{The Proposed GDP}\label{sec:def}
 
 Let the nodes $V$ of DAG $D=(V,E)$ now represent the group-specific random probability measures $G_j$'s. The edges $E$ represent the conditional dependence of $G_j$'s. Then the joint distribution of the random probability measures follows the DAG factorization  $\mathcal{P}(G_1,\dots,G_p|D) = \prod_{j = 1}^{p}\mathcal{P}(G_j \mid G_{pa(j)})$, where $G_{pa(j)}$ is the set of random probability measures indexed by the parents $pa(j)$ of  node $j$. For convenience, we assume $D$ has a unique root; see Figure \ref{fig::GDP introduction}. This assumption does not diminish the generality of our approach as a DAG with multiple roots can always be converted, without losing any conditional dependencies, to a DAG with a unique root by simply augmenting the DAG with a hidden common parent of the roots; that hidden common parent becomes the unique root of the new DAG (Figure \ref{fig::DAG with no unique ancestor}). The augmentation only changes the Markov blanket of the original root nodes. Specifically, the Markov blanket of any original root node is simply augmented with the hidden parent node. 
 As the Markov blanket of any other node remains unchanged, the distributions of all other nodes remain the same, and hence this augmentation does not alter the conditional dependencies of the original DAG.

Let us introduce a few terms before describing the proposed GDP. We denote the root node, which may be hidden, as the \emph{layer 0} of DAG $D$. The child nodes of the root node are termed as the \emph{layer-1} nodes, and we assume that there are $l_1$ of them. Similarly, we assume that there are a total of $l_2$ child nodes from the layer-1 nodes, which we refer to as the \emph{layer-2} nodes. We assume that there are $K$ layers in the given DAG $D$ and at any layer $k$, there are $l_k$ nodes. The total number of non-root nodes is  $\sum_{k=1}^{K}l_k = p$.  We define the concentration parameters and random measures of node $j$ in the layer $k$ of DAG $D$ as $\alpha_j^{(k)}$ and $G_j^{(k)}, \ \ j = 1, \dots, l_k$. We denote by $an^{(k, l)}(j)$ the collection of generation-$l$ ancestors of node $j$ in layer $k$ of the DAG. For example, $an^{(k, 1)}(j)$ denotes the parents (generation-$1$ ancestors) of the node $j$ in layer $k$, and $an^{(k, 2)}(j)$ denotes the collection of the parents of the nodes in $an^{(k, 1)}(j)$ or in other words, $an^{(k, 2)}(j)$ denotes the collection of ``grand-parents" (generation-$2$ ancestors) of node $j$ in layer $k$ of the DAG.

 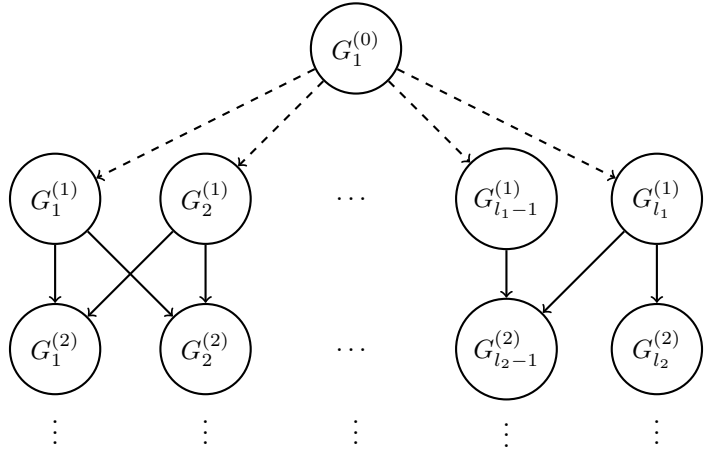
\begin{figure}[htp]
\begin{center}
\begin{tikzpicture}[node distance={20mm}, thick, main/.style = {draw, circle, minimum size=1.2cm}] 
\node[main, draw = none] (1) {}; 
\node[main, draw = none] (2) [right of = 1] {}; 
\node[main] (3) [right of=2] {$G_1^{(0)}$}; 
\node[main, draw = none] (3h1) [right of=3] {}; 
\node[main, draw = none] (3h2) [right of=3h1] {}; 
\node[main] (4) [below of=1] {$G_1^{(1)}$}; 
\node[main] (5) [right of=4] {$G_2^{(1)}$}; 
\node[main, draw = none] (6) [right of=5] {}; 
\node[main] (7) [below of=3h1] {$G_{l_1-1}^{(1)}$};
\node[main] (8) [below of=3h2] {$G_{l_1}^{(1)}$}; 
\node[main] (9) [below of=4] {$G_1^{(2)}$}; 
\node[main] (10) [right of=9] {$G_2^{(2)}$};
\node[main, draw = none] (11) [right of=10] {};
\node[main] (12) [below of=7] {$G_{l_2-1}^{(2)}$};
\node[main] (13) [below of=8] {$G_{l_2}^{(2)}$};

\draw[->, dashed] (3) -- (4); 
\draw[->, dashed] (3) -- (5);

\draw[->, dashed] (3) -- (7);
\draw[->, dashed] (3) -- (8);
\draw[->] (4) -- (9);
\draw[->] (4) -- (10);
\draw[->] (5) -- (9);
\draw[->] (5) -- (10);

\draw[->] (7) -- (12);
\draw[->] (8) -- (13);
\draw[->] (8) -- (12);

\path (5) -- node[auto=false]{\ldots} (7);
\path (10) -- node[auto=false]{\ldots} (12);
\node[main, draw = none] (14) [below of=9] {}; 
\node[main, draw = none] (15) [below of=10] {}; 
\node[main, draw = none] (16) [below of=11] {}; 
\node[main, draw = none] (17) [below of=12] {}; 
\node[main, draw = none] (18) [below of=13] {}; 
\path (9) -- node[auto=false]{$\vdots$} (14);
\path (10) -- node[auto=false]{$\vdots$} (15);
\path (11) -- node[auto=false]{$\vdots$} (16);
\path (12) -- node[auto=false]{$\vdots$} (17);
\path (13) -- node[auto=false]{$\vdots$} (18);
\end{tikzpicture} 
\caption{DAG augmented with a hidden root $G_1^{(0)}$, indicated by the dashed arrows. The original root nodes are $G_1^{(1)},\dots,G_{l_1}^{(1)}$. }
\label{fig::DAG with no unique ancestor}
\end{center}
\end{figure}
  We define GDP recursively from layer 0, the root node, 
\begin{equation}
\label{gdp distribution eq 1}
    G_1^{(0)} \mid \alpha_1^{(0)}, G_0 \sim DP\left(\alpha_1^{(0)}, G_0\right),
\end{equation} where $G_0$ is a fixed base probability measure. 
Then the distribution of the random probability measure of node $j$ in layer $k$ of DAG $D$ conditional on the concentration parameters and random probability measures of its parent nodes is given by, 
\begin{equation}
\label{gdp distribution eq 2}
    G_j^{(k)} \mid \alpha_j^{(k)}, \{G_l^{(k-1)} : l \in an^{(k, 1)}(j)\}  \sim DP\left(\alpha_j^{(k)}, \sum_{l\in an^{(k, 1)}(j)} \pi_{jl}^{(k)} G_l^{(k-1)}\right),
\end{equation}
for $j = 1,2, \dots, l_k$.
In other words, node $j$ in layer $k$ of the DAG is distributed according to a DP with its own concentration parameter $\alpha_j^{(k)}$ and its base distribution being a weighted average of the random probability measures of its parents in layer $k-1$ of the DAG, $\{G_l^{(k-1)} : l \in an^{(k, 1)}(j)\}$, where the weights are given by $\{\pi_{jl}^{(k)}: l \in an^{(k, 1)}(j)\}$, which have a unit sum $\sum_{l\in an^{(k, 1)}(j)} \pi_{jl}^{(k)}  = 1$. Moreover, from the Markov properties of DAG $D$, $G_{j_1}^{(k)}$ and $G_{j_2}^{(k)}$
are conditionally independent given their parents, $\{G_l^{(k-1)} : l \in an^{(k, 1)}(j_1)\}$ and/or $\{G_l^{(k-1)} : l \in an^{(k, 1)}(j_2)\}$, and $G_j^{(k)}$ is conditionally independent of all other random probability measures given its Markov blanket.

We remark that HDP is a special case of the proposed GDP with a specific DAG, fork-DAG (Figure \ref{fig::HDP introduction}).
Using the notations introduced, a fork-DAG is a DAG with a unique root node and only one layer of $l_1$ child nodes. With this specific DAG, the GDP is given by 
\begin{align*}
    &G_1^{(0)} \mid \alpha_1^{(0)}, G_0 \sim DP\left(\alpha_1^{(0)}, G_0\right),\\
    &G_j^{(1)} \mid \alpha_j^{(1)}, G_1^{(0)}  \sim DP\left(\alpha_j^{(1)}, G_1^{(0)} \right), \hspace{0.2cm} j = 1,2, \dots, l_1,
\end{align*}
which is clearly an HDP.

\subsection{GDP mixture model}
To cluster observations that are organized into possibly non-exchangeable groups, we use the proposed GDP in Section \ref{sec:def} as a mixing distribution of a mixture model.  
Letting $j$ index the groups and $i$ index the observations within each group, we assume that the observations $x_{j1}, x_{j2}, \dots, x_{jn_j}$ are exchangeable within each group $j$ but the groups may not be exchangeable.
We assume that each observation within a group is drawn independently from the mixture model \eqref{hdp mixture model} and $G_j$'s follow the GDP \eqref{gdp distribution eq 1} and \eqref{gdp distribution eq 2}.

\subsection{Hyperpriors}\label{sec:hyperpriors}
We assign a Dirichlet prior on the weights $\{\pi_{jl}^{(k)}\, :\, l \in an^{(k, 1)}(j)\}$ in \eqref{gdp distribution eq 2},
\begin{equation}
    \{\pi_{jl}^{(k)}\, :\, l \in an^{(k, 1)}(j)\}  \sim Dir\left(\{\alpha_l^{(k-1)} : l \in an^{(k, 1)}(j)\}\right),
\end{equation} 
where the parameters $\{\alpha_l^{(k-1)} : l \in an^{(k, 1)}(j)\}$ correspond to the concentration parameters of the parents (generation-$1$ ancestors) of node $j$. Since the concentration parameter of a DP relates to its precision (inverse-variance), assuming a Dirichlet prior for the mixture weights of any node with Dirichlet parameters proportional to the precisions of the parent nodes is a natural choice. This gives more ``weightage" to a parent node with a higher precision as opposed to a parent node with a lower precision.

The other distributional consideration that significantly simplifies the distribution of the random measure of any particular node is by considering a gamma-DAG distribution on the concentration parameters $\alpha_j^{(k)}$'s, which, like the distribution of $G_j^{(k)}$'s, also respects the same Markov property of DAG $D$. 
Specifically, we assume that
\begin{equation}
\begin{aligned}
  \label{eq:alpha_prior}
    &\alpha_1^{(0)} \mid \alpha_0 \sim Gamma(\alpha_0, 1),\\ 
    &\alpha_j^{(k)} \mid \{\alpha_l^{(k-1)} : l \in an^{(k, 1)}(j)\} \sim Gamma\left(\sum_{l\in an^{(k, 1)}(j)} \alpha_l^{(k-1)}, 1\right), &\hspace{0.2cm} j = 1,2, \dots, l_k.
\end{aligned}
\end{equation}
In other words, the concentration parameter of the root node follows a gamma distribution with a fixed shape $\alpha_0$ and a unit rate. The concentration parameter at any level of the DAG follows a conditionally gamma distribution with the shape parameter equal to the sum of the shape parameters of its parents. Such a choice of Gamma hyperprior on the concentration parameters of bottom level DPs of HDP have been considered in \citealp{williamson2013parallel}. We extend such a construction for the more general framework of our proposed GDP. In the next section, we will see how our choice of hyperpriors and hyperparameters leads to several compact representations of the proposed GDP, which requires two lemmas. The first lemma is Lemma 3.1 from \citealp{Sethuraman}, which we state here.
\begin{lemma}[\citealp{Sethuraman}]
\label{lemmaI}
\sloppy Let $\boldsymbol{\alpha}_1 = (\alpha_{11}, \alpha_{12}, \dots, \alpha_{1k})$ and $\boldsymbol{\alpha}_2 = (\alpha_{21}, \alpha_{22}, \dots, \alpha_{2k})$ be $k$-dimensional vectors with $\alpha_{ij} > 0 \hspace{0.2cm} \forall \hspace{0.2cm} j = 1, 2, \dots ,k,\hspace{0.2cm} i = 1, 2$. Let $\boldsymbol{X}_1$ and $\boldsymbol{X}_2$ be independent $k$-dimensional random vectors distributed as Dirichlet distribution with parameters $\boldsymbol{\alpha}_1$ and $\boldsymbol{\alpha}_2$, respectively. Let $\alpha_{1 \cdot} = \sum_{j=1}^{k} \alpha_{1j}$ and $\alpha_{2\cdot} = \sum_{j=1}^{k} \alpha_{2j}$. Let $\pi$ be independent of $\boldsymbol{X}_1$ and $\boldsymbol{X}_2$ and have a beta distribution $Beta\left(\alpha_{1\cdot}, \alpha_{2\cdot}\right)$. Then the distribution of $\pi \boldsymbol{X}_1 + \left(1 -\pi\right)\boldsymbol{X}_2$ is the Dirichlet distribution with parameter $\boldsymbol{\alpha}_1 + \boldsymbol{\alpha}_2$.
\end{lemma} 
The proof is provided in Section \ref{supp-sec::proof_lemma} of the Supplementary Materials for completeness. The next lemma is an immediate extension of \cref{lemmaI} for more than two independent Dirichlet distributed random vectors.
As the Dirichlet distribution is a multivariate analog of the beta distribution, by considering a Dirichlet distribution on the weights, we arrive at a similar result. This lemma is a finite-dimensional version of Theorem 1 of \citealp{williamson2013parallel}, which essentially states that a finite Dirichlet mixture of DPs is, in turn, a DP with its concentration parameter being the sum of the concentration parameters of the component DPs, and the base measure being a weighted mixture of the corresponding mixing base measures. 
\begin{lemma}
\label{lemmaII}
\sloppy Let $\boldsymbol{\alpha}_1, \boldsymbol{\alpha}_2, \dots, \boldsymbol{\alpha}_L$ be $k$-dimensional vectors where $\boldsymbol{\alpha}_i = (\alpha_{i1},\dots, \alpha_{ik})$ with $\alpha_{ij} > 0 \hspace{0.2cm} \forall \hspace{0.2cm} j = 1, 2, \dots ,k$,  $i = 1,2, \dots, L$. Let $\boldsymbol{X}_1, \boldsymbol{X}_2, \dots, \boldsymbol{X}_L$ be independent $k$-dimensional random vectors distributed as Dirichlet distribution with parameters $\boldsymbol{\alpha}_1, \boldsymbol{\alpha}_2,\dots,  \boldsymbol{\alpha}_L$, respectively. Let $\alpha_{i \cdot} = \sum_{j=1}^{k} \alpha_{ij},\ \ i = 1,2,\dots, L$. Let $\boldsymbol{\pi} = \left(\pi_1, \pi_2, \dots, \pi_L\right)$ be independent of $\boldsymbol{X}_1, \boldsymbol{X}_2,\dots, \boldsymbol{X}_L$ and have a Dirichlet distribution $Dir\left(\alpha_{1\cdot}, \alpha_{2\cdot},\dots,  \alpha_{L\cdot}\right)$. Then the distribution of $\sum_{i=1}^{L}\pi_i \boldsymbol{X}_i$ is the Dirichlet distribution with parameter $\sum_{i=1}^{L}\boldsymbol{\alpha}_i$. 
\end{lemma}
The proof is provided in Section \ref{supp-sec::proof_lemma} of the Supplementary Materials.

\section{Representations of the Graphical Dirichlet Process}\label{sec:representations}
In this section, we characterize the proposed GDP through (i) the hypergraph representation, (ii) the stick-breaking representation, (iii) the restaurant-type process representation, and (iv) the limit of finite mixture representation.

\subsection{The hypergraph representation}
The GDP, along with the hyperpriors on the concentration parameters and mixture weights, can be represented hierarchically as,
\begin{align}
\begin{aligned}
\label{gdp equations}
    & \alpha_1^{(0)} \mid \alpha_0 \sim Gamma\left(\alpha_0, 1\right),\\
    & G_1^{(0)} \mid \alpha_1^{(0)}, G_0 \sim DP\left(\alpha_1^{(0)}, G_0\right),\\
    & \alpha_j^{(k)} \mid \{\alpha_l^{(k-1)} : l \in an^{(k, 1)}(j)\} \sim Gamma\left(\sum_{l\in an^{(k, 1)}(j)} \alpha_l^{(k-1)}, 1\right),\\
    & \{\pi_{jl}^{(k)}\, :\, l \in an^{(k, 1)}(j)\} \mid \{\alpha_l^{(k-1)} : l \in an^{(k, 1)}(j)\} \sim Dir\left(\{\alpha_l^{(k-1)} : l \in an^{(k, 1)}(j)\}\right), \\
    &  G_j^{(k)} \mid \alpha_j^{(k)}, \{G_l^{(k-1)} : l \in an^{(k, 1)}(j)\}  \sim DP\left(\alpha_j^{(k)}, \sum_{l\in an^{(k, 1)}(j)} \pi_{jl}^{(k)} G_l^{(k-1)}\right),  
    \end{aligned}
\end{align}
for $j = 1,2, \dots, l_k$ and $k=1,\dots,K$.

The hyperparameters of the GDP consist of the base probability measure $G_0$ and the concentration parameter $\alpha_0$. 
The probability measure $G_1^{(0)}$ of the root node varies around the base measure $G_0$ with the amount of variability governed by $\alpha_1^{(0)}$, which in turn is governed by the hyperparameter $\alpha_0$. We now present a novel hypergraph representation of GDP, which simplifies the graph-based distribution. The representation follows from the gamma-DAG distribution on the concentration parameters and standard properties of Dirichlet distribution. 

\begin{theorem}[Hypergraph Representation]
\label{main theorem}
 Consider a DAG $D$ that has $K$ layers and $l_k$ distinct nodes in layer $k$ for $k=1,\dots,K$.
Under model \eqref{gdp equations}, the distribution of the random measure $G_j^{(k)}$ of node $j$ in layer $k$ of DAG $D$ can be equivalently represented as,
\begin{align*}
    & G_j^{(k)} \mid  \alpha_j^{(k)}, H_j^{(k,k)} \sim DP\left(\alpha_j^{(k)},  H_j^{(k,k)}\right),\\
    & H_j^{(k,k)} \mid \{\alpha_l^{(k-1)} : l \in an^{(k, 1)}(j)\}, H_j^{(k,k-1)} \sim DP\left(\sum_{l\in an^{(k, 1)}(j)} \alpha_l^{(k-1)}, H_j^{(k,k-1)}\right),\\
    & H_j^{(k,k-1)} \mid \{\alpha_l^{(k-2)} : l \in an^{(k, 2)}(j)\}, H_j^{(k,k-2)} \sim DP\left(\sum_{l\in an^{(k, 2)}(j)} \alpha_l^{(k-2)}, H_j^{(k,k-2)}\right),\\
    & \vdots \\
    & H_j^{(k,2)} \mid \{\alpha_l^{(1)} : l \in  an^{(k, k-1)}(j)\}, G_1^{(0)} \sim DP\left(\sum_{l\in an^{(k, k-1)}(j)} \alpha_l^{(1)}, G_1^{(0)}\right).
\end{align*}
\end{theorem}

The proof is provided in the Appendix \ref{appendix_GDP_example}. In words, Theorem \ref{main theorem} essentially states the following. The distribution of $G_j^{(k)}$  is a DP with a hidden base measure $H_j^{(k,k)}$ and the concentration parameter $\alpha_j^{(k)}$. 
The hidden base measure $H_j^{(k,k)}$, in turn, is again a DP with base measure $H_j^{(k,k-1)}$ and concentration parameter being the sum of the concentration parameters of the generation-$1$ ancestors of $G_j^{(k)}$. Recursively, the hidden base measure $H_j^{(k,k-1)}$ is a DP with base measure $H_j^{(k,k-2)}$ and the concentration parameter being the sum of the concentration parameters of the generation-$2$ ancestors. This distributional pattern continues in a hierarchical fashion. Through $k-1$ hidden base measures, any node in layer $k$ can be seen to depend on the root node $G_1^{(0)}$ through its ancestral relationships. We call the representation of GDP in Theorem \ref{main theorem} as the hypergraph representation because one can view $H_j^{(k,k-a)}$ for $a=0,\dots,k-2$ as a hypernode that contains all the sufficient information from generation-$(a+1)$ ancestors of $G_j^{(k)}$.
We provide in Figure \ref{DAG for random measures with hypernodes} an illustrative example of the hypergraph representation showing how the hypernodes contain all the ancestral information.  
From Figure \ref{fig::Hypernode H2}, we can see that the distribution of $G_6$ depends on the distribution of its parents, $G_2$ and $G_3$. We refer to $H_2$, consisting of $\{G_2, G_3\}$, as a hypernode. Hypernode $H_2$ contains all the information about the parents of $G_6$. Loosely speaking, the information of the root node $G_1$ (e.g., its atoms) is passed to $G_6$ through $H_2$. 
Similarly, $H_3$, being the hypernode of $\{G_3,G_4\}$, contains all the information about $G_7$ from its parent nodes allowing the flow of information from the root node (see Figure \ref{fig::Hypernode H3}). For node $G_8$, we have two levels of hypernodes -- $H_4$ denotes the first layer and consists of the parents of $G_8$, and $H^*$ denotes the second layer and consists of generation-$2$ ancestors of $G_8$. Thus, hypernodes $H_4$ and $H^*$ carry all the information from the root node $G_1$ to $G_8$ as illustrated in Figure \ref{fig::Hypernode H4 and H*}. 

\begin{figure}[htp]
\begin{center}
\begin{subfigure}[t]{0.32\textwidth}
        \centering
\begin{tikzpicture}[node distance={15mm}, thick, main/.style = {draw, circle}, scale=0.75, transform shape] 
\node[main, rectangle] (0) {$G_0$}; 
\node[main] (1) [below of=0] {$G_1$}; 
\node[main] (3) [below of=1] {$G_3$};
\node[main] (2) [left of=3] {$G_2$}; 
\node[ellipse, dashed, draw=black, fit=(2) (3), inner sep=0.5mm, label=above left:{$H_2$}] (all) {};
\node[main] (4) [right of=3] {$G_4$}; 
\node[main] (5) [below of=2] {$G_5$}; 
\node[main] (6) [below of=3] {$G_6$};
\node[main] (7) [below of=4] {$G_7$};
\node[main] (8) [below of=6] {$G_8$};
\draw[->] (0) -- (1); 
\draw[->] (1) -- (2); 
\draw[->] (1) -- (3); 
\draw[->] (1) -- (4);
\draw[->] (2) -- (5); 
\draw[->] (2) -- (6);
\draw[->] (3) -- (6); 
\draw[->] (3) -- (7); 
\draw[->] (4) -- (5);
\draw[->] (4) -- (7);
\draw[->] (5) -- (8);
\draw[->] (6) -- (8);
\draw[->] (7) -- (8);
\end{tikzpicture} 
\caption{}
\label{fig::Hypernode H2}
\end{subfigure}
\hspace{1mm}
\begin{subfigure}[t]{0.32\textwidth}
        \centering
\begin{tikzpicture}[node distance={15mm}, thick, main/.style = {draw, circle}, scale=0.75, transform shape] 
\node[main, rectangle] (0) {$G_0$}; 
\node[main] (1) [below of=0] {$G_1$}; 
\node[main] (3) [below of=1] {$G_3$};
\node[main] (2) [left of=3] {$G_2$}; 
\node[main] (4) [right of=3] {$G_4$}; 
\node[ellipse, dashed, draw=black, fit=(3) (4), inner sep=0.5mm, label=above right:{$H_3$}] (all) {};
\node[main] (5) [below of=2] {$G_5$}; 
\node[main] (6) [below of=3] {$G_6$};
\node[main] (7) [below of=4] {$G_7$};
\node[main] (8) [below of=6] {$G_8$};
\draw[->] (0) -- (1); 
\draw[->] (1) -- (2); 
\draw[->] (1) -- (3); 
\draw[->] (1) -- (4);
\draw[->] (2) -- (5); 
\draw[->] (2) -- (6);
\draw[->] (3) -- (6); 
\draw[->] (3) -- (7); 
\draw[->] (4) -- (5);
\draw[->] (4) -- (7);
\draw[->] (5) -- (8);
\draw[->] (6) -- (8);
\draw[->] (7) -- (8);
\end{tikzpicture} 
\caption{}
\label{fig::Hypernode H3}
\end{subfigure}
\hspace{1mm}
\begin{subfigure}[t]{0.32\textwidth}
        \centering
\begin{tikzpicture}[node distance={15mm}, thick, main/.style = {draw, circle}, scale=0.75, transform shape] 
\node[main, rectangle] (0) {$G_0$}; 
\node[main] (1) [below of=0] {$G_1$}; 
\node[main] (3) [below of=1] {$G_3$};
\node[main] (2) [left of=3] {$G_2$}; 
\node[main] (4) [right of=3] {$G_4$}; 
\node[ellipse, dashed, draw=black, fit=(2) (3) (4), inner sep=0.5mm, label=right:{$H^*$}] (all) {};
\node[main] (5) [below of=2] {$G_5$}; 
\node[main] (6) [below of=3] {$G_6$};
\node[main] (7) [below of=4] {$G_7$};
\node[ellipse, dashed, draw=black, fit=(5) (6) (7), inner sep=0.5mm, label=right:{$H_4$}] (all) {};
\node[main] (8) [below of=6] {$G_8$};
\draw[->] (0) -- (1); 
\draw[->] (1) -- (2); 
\draw[->] (1) -- (3); 
\draw[->] (1) -- (4);
\draw[->] (2) -- (5); 
\draw[->] (2) -- (6);
\draw[->] (3) -- (6); 
\draw[->] (3) -- (7); 
\draw[->] (4) -- (5);
\draw[->] (4) -- (7);
\draw[->] (5) -- (8);
\draw[->] (6) -- (8);
\draw[->] (7) -- (8);
\end{tikzpicture} 
\caption{}
\label{fig::Hypernode H4 and H*}
\end{subfigure}
\caption{Illustration of hypernodes (represented by dashed ovals) of the DAG for our motivational problem. (a) Hypernode $H_2$ consists of the generation-1 ancestors (i.e., $G_2$ and $G_3$) of node $G_6$. (b) Hypernode $H_3$ consists of the generation-1 ancestors (i.e., $G_3$ and $G_4$) of node $G_7$. (c) Hypernode $H_4$ consists of the generation-1 ancestors (i.e., $G_5$, $G_6$, and $G_7$) of node $G_8$. Hypernode $H^*$ consists of the generation-2 ancestors (i.e., $G_2$, $G_3$, and $G_4$) of node $G_8$.}
\label{DAG for random measures with hypernodes}
\end{center}
\end{figure}
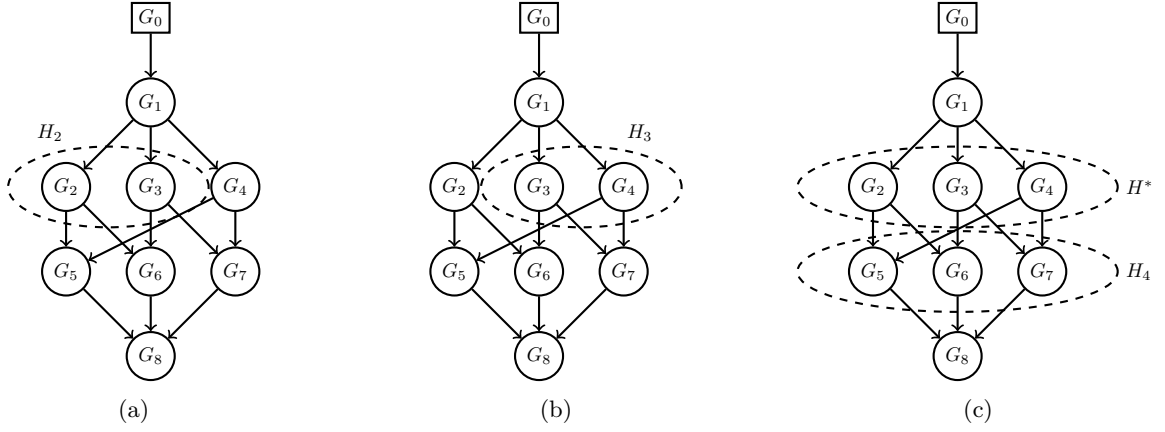

We will exploit this representation to derive the stick-breaking representation and the limit of finite mixture representation of the proposed GDP in the next subsections.

\subsection{The stick-breaking representation}
Given that the random measure $G_1^{(0)}$ of the root node is distributed as a DP, it can be expressed using a stick-breaking representation, 
\begin{equation}
    G_1^{(0)} = \sum_{l=1}^{\infty} \beta_{1l}^{(0)} \delta_{\phi_l},
\end{equation}
where $\phi_l \overset{iid}{\sim} G_0$ and $\boldsymbol{\beta}_1^{(0)} = \left(\beta_{1l}^{(0)}\right)_{l=1}^{\infty}  \sim \text{GEM}\left(\alpha_1^{(0)}\right)$ are mutually independent. We interpret $\boldsymbol{\beta}_1^{(0)}$ as a probability measure on the positive integers. Since $G_1^{(0)}$ has support at the atoms $\boldsymbol{\phi} = (\phi_{l})_{l=1}^{\infty} $, each $G_j^{(k)}$ necessarily has support at these atoms as well and hence can be expressed as, 
\begin{equation}
        G_j^{(k)} = \sum_{l = 1}^{\infty} \beta_{jl}^{(k)} \delta_{\phi_l}.
\end{equation}
As with Theorem \ref{main theorem}, the stick-breaking weights depend hierarchically on a set of hidden weights. Letting $\boldsymbol{\beta}_{j}^{(k)} = \left(\beta_{jl}^{(k)}\right)_{l=1}^{\infty}$ be the stick-breaking weights for node $j$ in layer $k$ of DAG $D$ and letting $\boldsymbol{\nu}_{j}^{(k, m)} = \left(\nu_{jl}^{(k,m)}\right)_{l=1}^{\infty}, m = 2,\dots, k$ be their hidden weights, we have the following corollary.
\begin{corollary}[Stick-Breaking Representation]
Consider a DAG $D$ that has $K$ layers and $l_k$ distinct nodes in layer $k$ for $k=1,\dots,K$. The stick-breaking weights $\boldsymbol{\beta}_j^{(k)}$ of node $j$ at layer $k$ of DAG $D$ can be represented as
\begin{align*}
    & \boldsymbol{\beta}_{j}^{(k)} \mid  \alpha_j^{(k)}, \boldsymbol{\nu}_j^{(k,k)} \sim DP\left(\alpha_j^{(k)},  \boldsymbol{\nu}_j^{(k,k)}\right),\\
    & \boldsymbol{\nu}_j^{(k,k)} \mid \{\alpha_l^{(k-1)} : l \in an^{(k, 1)}(j)\}, \boldsymbol{\nu}_j^{(k, k-1)} \sim DP\left(\sum_{l\in an^{(k, 1)}(j)} \alpha_l^{(k-1)}, \boldsymbol{\nu}_j^{(k,k-1)}\right),\\
    & \boldsymbol{\nu}_j^{(k,k-1)} \mid \{\alpha_l^{(k-2)} : l \in an^{(k, 2)}(j)\}, \boldsymbol{\nu}_j^{(k,k-2)} \sim DP\left(\sum_{l\in an^{(k, 2)}(j)} \alpha_l^{(k-2)}, \boldsymbol{\nu}_j^{(k,k-2)}\right),\\
    & \vdots \\
    & \boldsymbol{\nu}_j^{(k,2)} \mid \{\alpha_l^{(1)} : l \in  an^{(k, k-1)}(j)\}, \boldsymbol{\beta}_1^{(0)} \sim DP\left(\sum_{l\in an^{(k, k-1)}(j)} \alpha_l^{(1)}, \boldsymbol{\beta}_1^{(0)}\right).
\end{align*}

\end{corollary}

The proof of this corollary directly follows from the hypergraph representation of Theorem \ref{main theorem} and is hence omitted. We call this representation \emph{the stick-breaking representation} where $\boldsymbol{\nu}_j^{(k, k)}$ is interpreted as a hidden probability measure on the set of positive integers corresponding to the first hidden layer. Each hidden layer of stick-breaking weights depend hierarchically on its previous hidden layer, denoted by $\boldsymbol{\nu}_j^{(k,k-1)}$, $\boldsymbol{\nu}_j^{(k,k-2)}$, and so on, and finally on the weights $\boldsymbol{\beta}_1^{(0)}$ of the root node.

\subsection{The family-owned restaurant process representation}\label{sec:forp}
DP and HDP have the well-known Chinese restaurant process and franchise representations. Here, we provide a culinary analog for the proposed GDP. We refer to this process as the \emph{family-owned restaurant process} as it is customary to use familial relationships to describe the relationships between nodes in a DAG. 
The metaphor is as follows. An original restaurant is opened by the ancestor of a family (the root node), which serves some dishes from a global menu containing an infinite number of dishes. The descendants of the ancestor open their own respective restaurants, which serve some of the dishes already being served in the restaurants owned by their parents and possibly some new dishes from the global menu. 
At each table of the original restaurant, one dish is ordered from the menu by the first customer occupying the table, and the dish is shared by all the other customers who sit at that table. Any subsequent customer may either join an occupied table and share the dish being served at that table or open a new table with a new dish from the menu. In restaurants other than the original restaurant, however, the first customer might choose to select a dish being served at one of the tables of its parent restaurant or order a new dish from the menu. Since the hypergraph representation of GDP involves hypernodes with hidden probability measures, we introduce a notation for the number of tables serving a dish in any restaurant and demarcate them with the notation for the number of tables serving the dish in the hypernodes, which we refer to as hyper-restaurants. 

As before, assume that there are $K$ generations in the family and there are $l_k$ different restaurants in generation $k$. The restaurants correspond to the nodes of DAG $D$. The customers coming in restaurant $j$ of generation $k$ correspond to parameters $\theta_{ji}^{(k)}$. Let $\phi_1, \phi_2, \dots, \phi_L$ denote i.i.d. random variables distributed according to the base distribution $G_0$, which are dishes from the global menu. 
To maintain a count of customers and tables, we introduce two notations. We use the notation $n_{jt}^{(k)}$ to denote the number of customers at table $t$ in the restaurant $j$ of generation $k$ and the notation $m_{jl}^{(k)}$ to denote the number of tables in the restaurant $j$ of generation $k$ that serve dish $l$. Marginal counts are represented by dots at the appropriate indices. For example, $m_{j\cdot}^{(k)}$ denotes the count of all the tables (regardless of what dishes being served) in the restaurant $j$ of generation $k$.
We introduce the notation $\psi_{jt}^{(k, k)}$ to denote the dish served at table $t$ in restaurant $j$ of generation $k$, chosen from the corresponding layer-$1$ hyper-restaurant ($H_j^{(k,k)})$. 

\par We integrate out random measures $\left\{G_j^{(k)}, H_j^{(k,k)}, H_j^{(k,k-1)}, \dots, G_1^{(0)}\right\}$ sequentially. First, we find the conditional distribution of $\theta_{ji}^{(k)}$ given $\theta_{j1}^{(k)}, \theta_{j2}^{(k)}, \dots, \theta_{j,i-1}^{(k)}, \alpha_j^{(k)}$, and $H_j^{(k,k)}$ with $G_j^{(k)}$ integrated out,
\begin{equation}
    \theta_{ji}^{(k)} \mid \theta_{j1}^{(k)}, \theta_{j2}^{(k)}, \dots, \theta_{j,i-1}^{(k)}, \alpha_j^{(k)}, H_j^{(k,k)} \sim \sum_{t = 1}^{m_{j\cdot}^{(k)}} \frac{n_{jt}^{(k)}}{i - 1 + \alpha_j^{(k)}}\delta_{\psi_{jt}^{(k,k)}} + \frac{\alpha_j^{(k)}}{i - 1 + \alpha_j^{(k)}} H_j^{(k,k)},
\end{equation}

We let $\psi_{jt}^{(k,k-1)}$ to denote the dish served at table $t$ in the layer-$1$ hyper-restaurant corresponding to restaurant $j$ of generation $k$, chosen from the dishes served in the layer-$2$ hyper-restaurants ($H_j^{(k,k-1)})$.
\sloppy Integrating out the hidden measure from the current layer $H_j^{(k,k)}$, the conditional distribution of $\psi_{jt}^{(k,k)}$ given $\psi_{j1}^{(k,k-1)}, \psi_{j2}^{(k,k-1)}, \dots, \psi_{j1}^{(k,k)}, \dots, \psi_{j,t-1}^{(k,k)}, \{\alpha_l^{(k-1)} : l \in an^{(k, 1)}(j)\}$, and the hidden measure from the previous layer,  $H_j^{(k,k-1)}$ is given by, 
\begin{multline}
        \psi_{jt}^{(k,k)} \mid \psi_{j1}^{(k,k-1)}, \psi_{j2}^{(k,k-1)}, \dots, \psi_{j1}^{(k,k)}, \dots, \psi_{j,t-1}^{(k,k)}, \{\alpha_l^{(k-1)} : l \in an^{(k, 1)}(j)\}, H_j^{(k,k-1)} \\
        \sim \mathlarger{\sum}_{l = 1}^{\scalebox{0.7}{$M_j^{{(k,1)}}$}} \frac{m_{jl}^{(k,1)}}{m_{j\cdot}^{(k, 1)} + \sum_{l\in an^{(k, 1)}(j)} \alpha_l^{(k-1)}}\mathlarger{\delta}_{\psi_{jl}^{(k,k-1)}} + \frac{\sum_{l\in an^{(k, 1)}(j)} \alpha_l^{(k-1)}}{m_{j\cdot}^{(k,1)} + \sum_{l\in an^{(k, 1)}(j)} \alpha_l^{(k-1)}} H_j^{(k,k-1)},
\end{multline}
where the notation $m_{j l}^{(k,1)}$ denotes the number of tables in layer-1 hyper-restaurant, corresponding to restaurant $j$ of generation $k$ serving the dish $l$. We denote by $M_j^{{(k,1)}}$ the number of dishes served in the layer-$1$ hyper-restaurants and by $m_{j\cdot}^{(k, 1)}$ the total number of tables in the layer-$1$ hyper-restaurant, corresponding to the restaurant $j$ of generation $k$. Similarly, integrating out the measure $H_j^{(k,k-1)}$ and introducing the next layer of variables $\psi_{jt}^{(k,k-2)}$, the conditional distribution of $\psi_{jt}^{(k, k-1)}$ given $\psi_{j1}^{(k,k-2)}, \psi_{j2}^{(k,k-2)}, \dots, \psi_{j1}^{(k, k-1)}, \dots, \psi_{j, t-1}^{(k,k-1)}, \{\alpha_l^{(k-2)} : l \in an^{(k, 2)}(j)\}$, and the hidden measure from the previous layer  $H_j^{(k,k-2)}$ is given by,
\begin{multline}
\psi_{jt}^{(k,k-1)} \mid \psi_{j1}^{(k,k-2)}, \psi_{j2}^{(k,k-2)}, \dots, \psi_{j1}^{(k,k-1)}, \dots, \psi_{j, t-1}^{(k,k-1)}, \{\alpha_l^{(k-2)} : l \in an^{(k, 2)}(j)\}, H_j^{(k,k-2)} \\
        \sim \mathlarger{\sum}_{l = 1}^{\scalebox{0.7}{$M_j^{{(k,2)}}$}} \frac{m_{jl}^{(k,2)}}{m_{j\cdot}^{(k,2)} + \sum_{l\in an^{(k, 2)}(j)} \alpha_l^{(k-2)}}\mathlarger{\delta}_{\psi_{jl}^{(k,k-2)}} + \frac{\sum_{l\in an^{(k, 2)}(j)} \alpha_l^{(k-2)}}{m_{j\cdot}^{(k,2)}+ \sum_{l\in an^{(k, 2)}(j)} \alpha_l^{(k-2)}} H_j^{(k,k-2)}.
\end{multline}
As in the stick-breaking representation, we can recursively integrate out hidden measures and eventually arrive at the conditional distribution of $\psi_{jt}^{(k, 2)}$ given $ \psi_{j1}^{(0)}, \psi_{j2}^{(0)}, \dots, \psi_{j1}^{(k, 2)}, \dots, \psi_{j,t-1}^{(k, 2)}, \{\alpha_l^{(1)} : l \in an^{(k, k-1)}(j)\}$, and the probability measure 
of the root node $G_1^{(0)}$,
\begin{multline}
        \psi_{jt}^{(k, 2)} \mid \psi_{j1}^{(0)}, \psi_{j2}^{(0)}, \dots, \psi_{j1}^{(k, 2)}, \dots, \psi_{j,t-1}^{(k, 2)}, \{\alpha_l^{(1)} : l \in an^{(k, k-1)}(j)\}, G_1^{(0)} \\
        \sim \mathlarger{\sum}_{l = 1}^{M_j^{(k, k-1)}} \frac{m_{jl}^{(k, k-1)}}{m_{j \cdot}^{(k, k-1)} + \sum_{l\in an^{(k, k-1)}(j)} \alpha_l^{(1)}}\mathlarger{\delta}_{\psi_{jl}^{(0)}} + \frac{\sum_{l\in an^{(k, k-1)}(j)} \alpha_l^{(1)}}{m_{j \cdot}^{(1)} + \sum_{l\in an^{(k, k-1)}(j)} \alpha_l^{(1)}} G_1^{(0)},
\end{multline}
and the conditional distribution of $\psi_{jt}^{(0)}$ given $\psi_{j1}^{(0)}, \dots, \psi_{j,t-1}^{(0)}, \alpha_1^{(0)}$, and the base measure $G_0$,
\begin{equation}\label{distribution_root_node}
        \psi_{jt}^{(0)} \mid \psi_{j1}^{(0)}, \dots, \psi_{j,t-1}^{(0)}, \alpha_1^{(0)}, G_0 
        \sim \mathlarger{\sum}_{l = 1}^{L} \frac{m_{l}^{(0)}}{m_{\cdot}^{(0)} + \alpha_1^{(0)}}\mathlarger{\delta}_{\phi_l} + \frac{\alpha_1^{(0)}}{m_{\cdot}^{(0)} + \alpha_1^{(0)}} G_0,
\end{equation}
where $m_{l}^{(0)}$ denotes the number of tables in the original restaurant serving dish $l$ and $m_{\cdot}^{(0)}$ denotes the total number of tables in the original restaurant. Note that \eqref{distribution_root_node} corresponds to the case where the root node is hidden (the same as in HDP). 
When the root node is not hidden, a similar formula can be derived, which is omitted for simplicity.

\subsection{The infinite limit of finite mixture model}
The GDP mixture model can  be derived as the infinite limit of a finite mixture model. Let us denote the observations and the mixture component indicator from node $j$ in layer $k$ of DAG $D$ by $x_{ji}^{(k)}$ and $z_{ji}^{(k)}$, respectively. Suppose $\boldsymbol{\beta}_1^{(0)}$ is the vector of mixing weights for the root node. Denoting by $\boldsymbol{\beta}_j^{(k)}$ the mixing weights of node $j$ in layer $k$ and by $\nu_j^{(k,m)}$ the corresponding mixing weights for the hidden layer $m$, with $m=2,\dots, k$, we consider a finite mixture version of the proposed GDP,
\begin{align}
\label{eq::finit_mixture_model}
   \nonumber \boldsymbol{\beta}_1^{(0)} \mid \alpha_1^{(0)} & \sim Dir\left(\alpha_1^{(0)}/L, \dots,  \alpha_1^{(0)}/L\right),\\
    \nonumber \boldsymbol{\nu}_j^{(k,2)} \mid \{\alpha_l^{(1)} : l \in an^{(k, k-1)}(j)\}, \boldsymbol{\beta}_1^{(0)} &\sim Dir\left(\sum_{l\in an^{(k, k-1)}(j)} \alpha_l^{(1)} \left(\beta_{11}^{(0)}, \dots, \beta_{1L}^{(0)}\right)\right),\\
    \nonumber & \vdots\\
    \nonumber \boldsymbol{\nu}_j^{(k,k)} \mid \{\alpha_l^{(k,k-1)} : l \in an^{(k, 1)}(j)\}, \boldsymbol{\nu}_j^{(k,k-1)} & \sim Dir\left(\sum_{l\in an^{(k, 1)}(j)} \alpha_l^{(k-1)} \left(\nu_{j1}^{(k,k-1)}, \dots, \nu_{jL}^{(k,k-1)}\right)\right),\\
    \nonumber \boldsymbol{\beta}_j^{(k)} \mid \alpha_j^{(k)}, \boldsymbol{\nu}_j^{(k,k)} & \sim Dir\left(\alpha_j^{(k)} \left(\nu_{j1}^{(k,k)}, \dots, \nu_{jL}^{(k,k)}\right)\right),\\
    \nonumber \phi_l \mid G_0 & \sim G_0, \\
    \nonumber z_{ji}^{(k)} \mid \boldsymbol{\beta}_j^{(k)} &\sim \boldsymbol{\beta}_j^{(k)}, \\
      x_{ji}^{(k)} \mid  z_{ji}^{(k)}, \left(\phi_l\right)_{l=1}^{L} &\sim F\left(\phi_{z_{ji}^{(k)}}\right).
\end{align} 
The distribution of this finite mixture model approaches the GDP mixture model as $L\rightarrow \infty$. Refer to Section \ref{supp-sec::proof_infinite_limit} of the Supplementary Materials for the proof. Based on this finite mixture model approximation with a large enough truncation level $L$, we develop an efficient posterior inference procedure of our model using a Metropolis-within-blocked-Gibbs sampler with a specialized proposal \citep{SALTSampler}; see Section \ref{supp-sec::fmm_posterior} of the Supplementary Materials for details.

\section{Simulations}\label{sec:simulations}
Our simulations are designed to mimic the motivating application where we have 8 experimental groups, whose relationships are represented by the DAG in Figure \ref{fig::DAG}.  

\begin{figure}[htp]
    \centering
\begin{tikzpicture}[node distance={15mm}, thick, main/.style = {draw, circle}] 
\node[main] (1) {$ 1$}; 
\node[main] (3) [below of=1] {$ 3$};
\node[main] (2) [left of=3] {$ 2$}; 
\node[main] (4) [right of=3] {$ 4$}; 
\node[main] (5) [below of=2] {$ 5$}; 
\node[main] (6) [below of=3] {$ 6$};
\node[main] (7) [below of=4] {$ 7$};
\node[main] (8) [below of=6] {$ 8$}; 
\draw[->] (1) -- (2); 
\draw[->] (1) -- (3); 
\draw[->] (1) -- (4);
\draw[->] (2) -- (5); 
\draw[->] (2) -- (6);
\draw[->] (3) -- (6); 
\draw[->] (3) -- (7); 
\draw[->] (4) -- (5);
\draw[->] (4) -- (7);
\draw[->] (5) -- (8);
\draw[->] (6) -- (8);
\draw[->] (7) -- (8);
\end{tikzpicture} 
\caption{The DAG of experimental groups.}
\label{fig::DAG}
\end{figure}
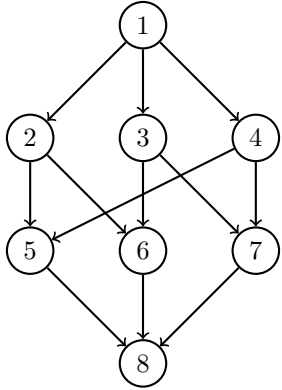

We generated data within each of the $8$ groups from a four-component mixture of bivariate Gaussian distributions with different covariance matrices for each group. We drew the DP concentration parameters $\alpha_j$'s for the different groups  from their prior distribution \eqref{eq:alpha_prior} respecting the DAG in Figure \ref{fig::DAG} with $\alpha_0=5$. The weights of the finite mixture model corresponding to the different groups were drawn using \eqref{eq::finit_mixture_model} and the same DAG.
The true cluster indicators of each group were sampled from a multinomial distribution with probabilities equal to the mixture weights. Using these true cluster indices for each group, samples were drawn from the Gaussian distribution with the cluster-specific mean and group-specific covariance matrix, given in Tables \ref{supp-table:cluster means_simulation} and \ref{supp-table:population covariance_simulation}, respectively, in Section \ref{supp-sec:simulations} of the Supplementary Materials. Refer to the same section in the Supplementary Materials for more details on our simulation strategy.
In our Gibbs sampler, the truncation level of the finite mixture model was set to $L = 10$, and the base measure for GDP, $G_0$, was specified as the normal-inverse-Wishart distribution, $\mathcal{NIW}(\boldsymbol{0}, 0.01, \mathbb{I}_2, 2)$.
Upon the completion of the Gibbs sampler, the clusters were estimated by using the least squares criterion \citep{LeastSquares}, and they were compared with the true cluster labels for evaluation. We considered various sample sizes in each group, which are summarized in Table \ref{tab:ss}.
In all cases, we ran 15,000 iterations of our Gibbs sampler and discarded the first 5,000 samples as burn-in.

\begin{table}[htp]
\centering
\begin{tabular}{ |c|c|c|c|c|c|c|c|c| }
 \hline
 Sample sizes & \multicolumn{8}{|c|}{Groups} \\ 
 \cline{2-9}
 & 1 & 2 & 3 &4 & 5 & 6 & 7 & 8\\
 \hline
 small & 40 & 30 & 30 & 35 & 25 & 30 & 25 & 30\\
moderate & 80 & 70 & 70 & 75 & 83 & 88 & 92 & 88\\
large & 150 & 160 & 180 & 170 & 155 & 175 & 185 & 145\\
 unbalanced & 350 & 30 & 40 & 45 & 25 & 25 & 35 & 35 \\
 \hline
\end{tabular}
\caption{The sample sizes for the different groups that were used to simulate the data.}
\label{tab:ss}
\end{table}
\begin{figure}[htp]
\centering
\begin{subfigure}{0.8\textwidth}
  \centering
  \includegraphics[width= \linewidth]{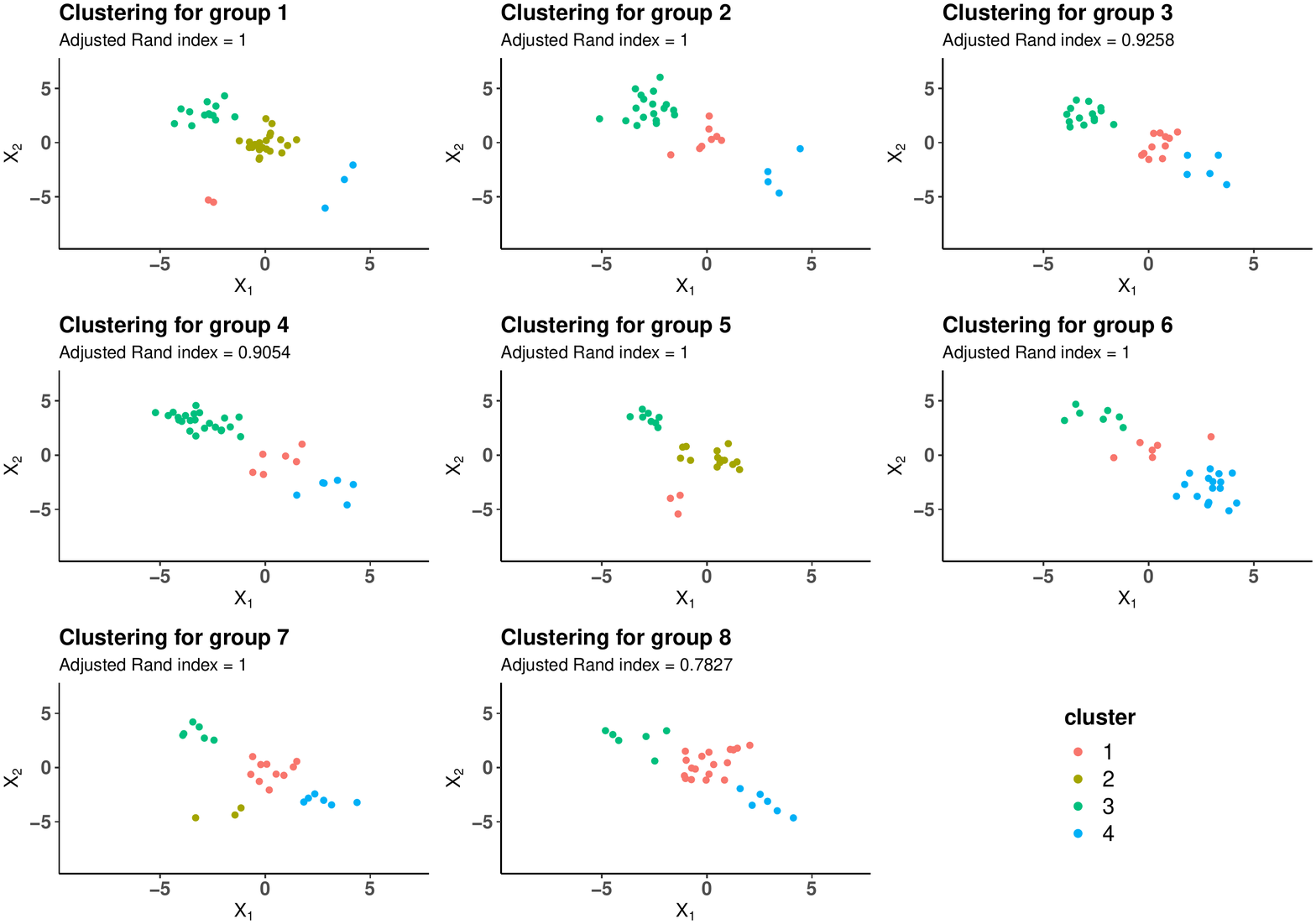}
  \caption{Small sample size in each group}
  \label{fig:clustering small}
\end{subfigure}%
\par\bigskip
\begin{subfigure}{.8\textwidth}
  \centering
  \includegraphics[width=\linewidth]{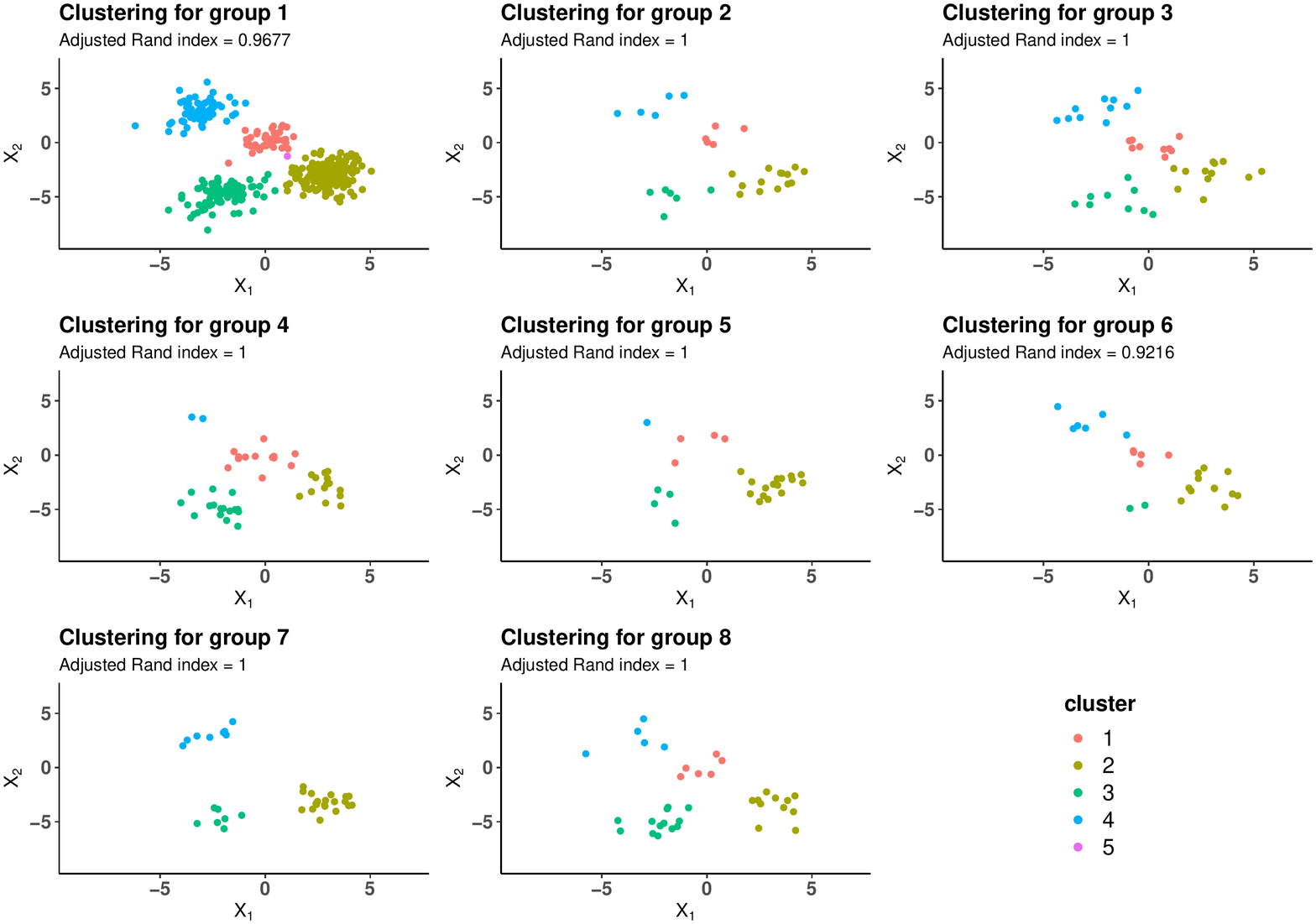}
  \caption{Unbalanced sample sizes between groups}
  \label{fig:clustering unbalanced}
\end{subfigure}
\caption{Clustering performance of GDP for different sample sizes. The colors indicate the estimated clusters by GDP.  Adjusted Rand index is reported at the top of each panel.}
\label{fig:clustering for various sample sizes}
\end{figure}

The clustering results of GDP for small and unbalanced sample sizes are visualized in Figure \ref{fig:clustering for various sample sizes}. The remaining clustering plots are shown in Supplementary \ref{supp-sec:simulations}. Across different sample sizes, the proposed GDP was able to identify the clusters within each group with very good accuracy and was able to link clusters across non-exchangeable groups. 
\par
We also looked at the clustering performance of GDP under a more difficult scenario. The simulation details and clustering results are shown in Supplementary \ref{supp-sec:simulations}. Since HDP is a special case of the proposed GDP, we compared the two methods for this difficult scenario. 
We also compared the clustering performance of GDP with k-means, a widely used non-Bayesian clustering technique. The number of clusters in k-means was taken to be the truncation level of our GDP. 
All simulations were replicated $50$ times. 

GDP significantly outperformed both HDP and k-means. For example, the boxplots of Adjusted Rand indices (\citealp{ARI}; higher is better) for the different methods are shown in Figure \ref{fig:gdp hdp comparison 1}.
It is evident that the Adjusted Rand indices of GDP were almost uniformly higher than those of HDP because HDP was not able to handle non-exchangeable groups. Similarly, the higher Adjusted Rand indices of GDP indicated its superior clustering performance over the k-means algorithm. Moreover, k-means algorithm does not allow sharing of relevant clusters across the groups. 

\begin{figure}[htp]
\centering
  \centering
  \includegraphics[width= 1\linewidth]{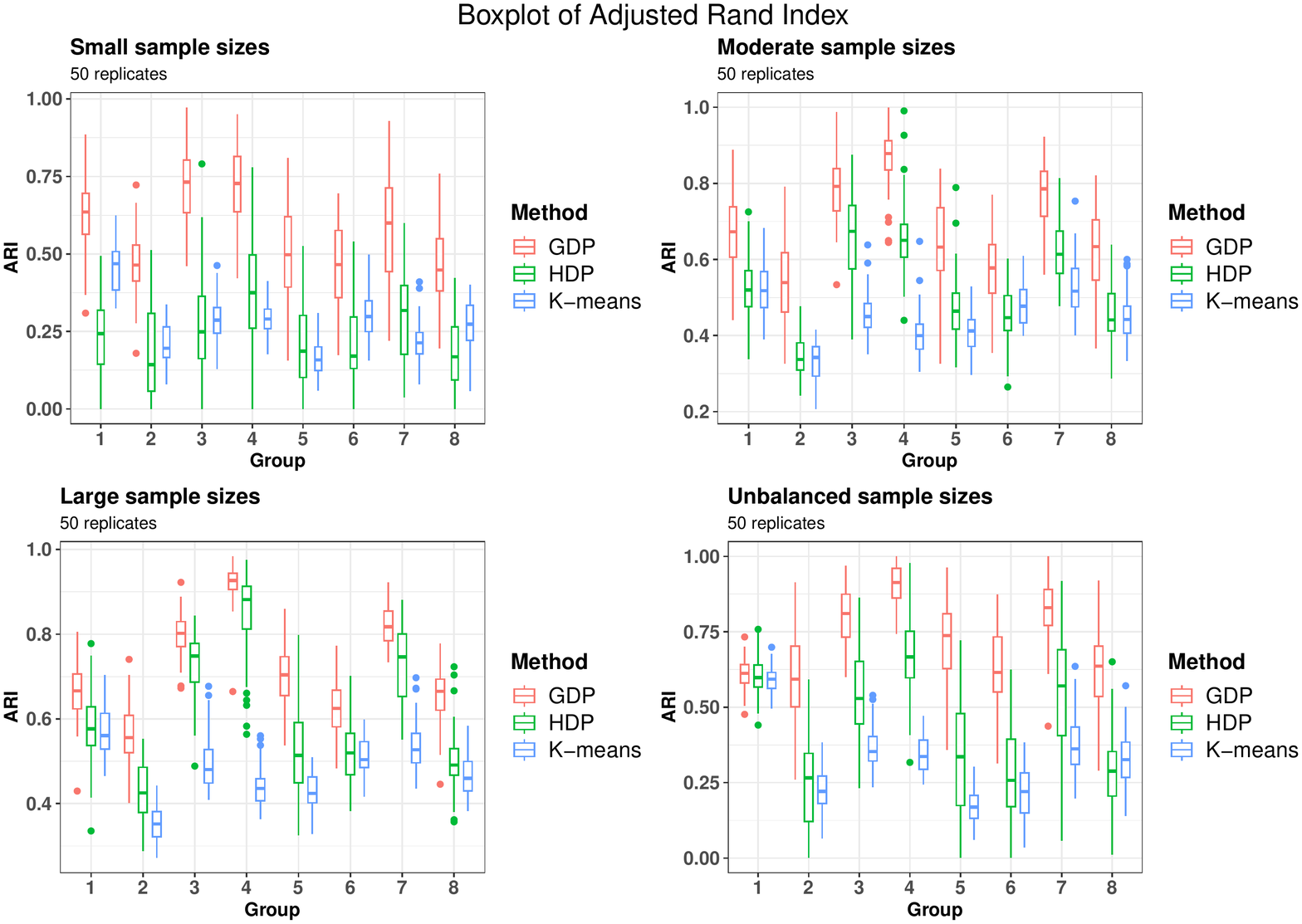}
\caption{The boxplots of the Adjusted Rand indices for GDP, HDP, and k-means for all sample sizes.}
\label{fig:gdp hdp comparison 1}
\end{figure}

\section{Real Data Analysis}\label{sec:read_data_analysis_main}
With the advancement of next-generation sequencing techniques in recent years, it is now possible to molecularly characterize individual cells, which may provide valuable insights into complex biological systems, ranging from cancer genomics to diverse microbial communities \citep{scRNA}. Colorectal cancer is the  third most common type of cancer after breast and lung cancers. It is known that the mutation of tumor-suppressor gene Apc is an initial step in most colorectal tumors \citep{apcmutation}. In addition, numerous studies have been conducted to understand the effect of high-fat vs low-fat diet on gene expressions (\citealp{fat_diet1}; \citealp{fat_diet2}; \citealp{fat_diet3}). We are motivated by a study that aimed to investigate how diet, genotype, and treatment with a new cancer prevention drug (AdipoRon) against placebo interacted to influence the expression of genes in intestinal crypt and tumor cells. The experiments started from a baseline group where the mice were genetically wild-type, fed with a normal diet, and treated with placebo. Then to understand the main effects of genotype, diet, and cancer treatment on stem cell gene expression, the experimenters introduced three new groups of mice, each differing from the baseline group 
by exactly one factor (Apc knock-out, high-fat diet, or new cancer treatment AdipoRon). To determine the two-way interaction effects, three additional groups of mice were studied, each of which differed from the baseline group 
by two factors (e.g., mice with Apc knock-out, high-fat diet, and placebo). Lastly, for a three-way interaction, the experimenters introduced the eighth group of mice with Apc knock-out, a high-fat diet, and the new treatment AdipoRon. By design, these 8 experimental groups are non-exchangeable and their relationships can be delineated by the DAG in Figure \ref{fig::DAG}. 
The goal of this analysis is to identify potential intestinal molecular subtypes within each experimental group while allowing information to be shared across these non-exchangeable groups with the proposed GDP model.
For illustration, we randomly sampled 100 cells from each of the eight groups.
The scRNA-seq data were pre-processed following standard procedure as outlined by \cite{seurat1} using the  R package \texttt{Seurat}. The data was log-normalized and scaled such that the mean expression across cells was $0$ and the variance across cells was $1$. 
As a common practice in single-cell data analyses, the uniform manifold approximation and projection (UMAP) \citep{umap} was used to reduce the data to two dimensions. We considered the same truncation level, $L = 10$, and the same base probability measure, $G_0$, as in the simulations.
We ran four parallel chains of the Gibbs sampler for $40,000$ iterations. To monitor the convergence of the sampler, we drew the traceplots of the log-likelihood for each of the four chains, after discarding the initial $25,000$ samples and thinning the samples by a factor of $15$, which indicated no lack of convergence of our sampler.
We pooled the Monte Carlo samples across different chains for posterior inference. We compared the clustering performance with that obtained from HDP on the same data.


\begin{figure}[htp]
\centering
\begin{subfigure}{0.8\textwidth}
  \centering
  \includegraphics[width= \linewidth]{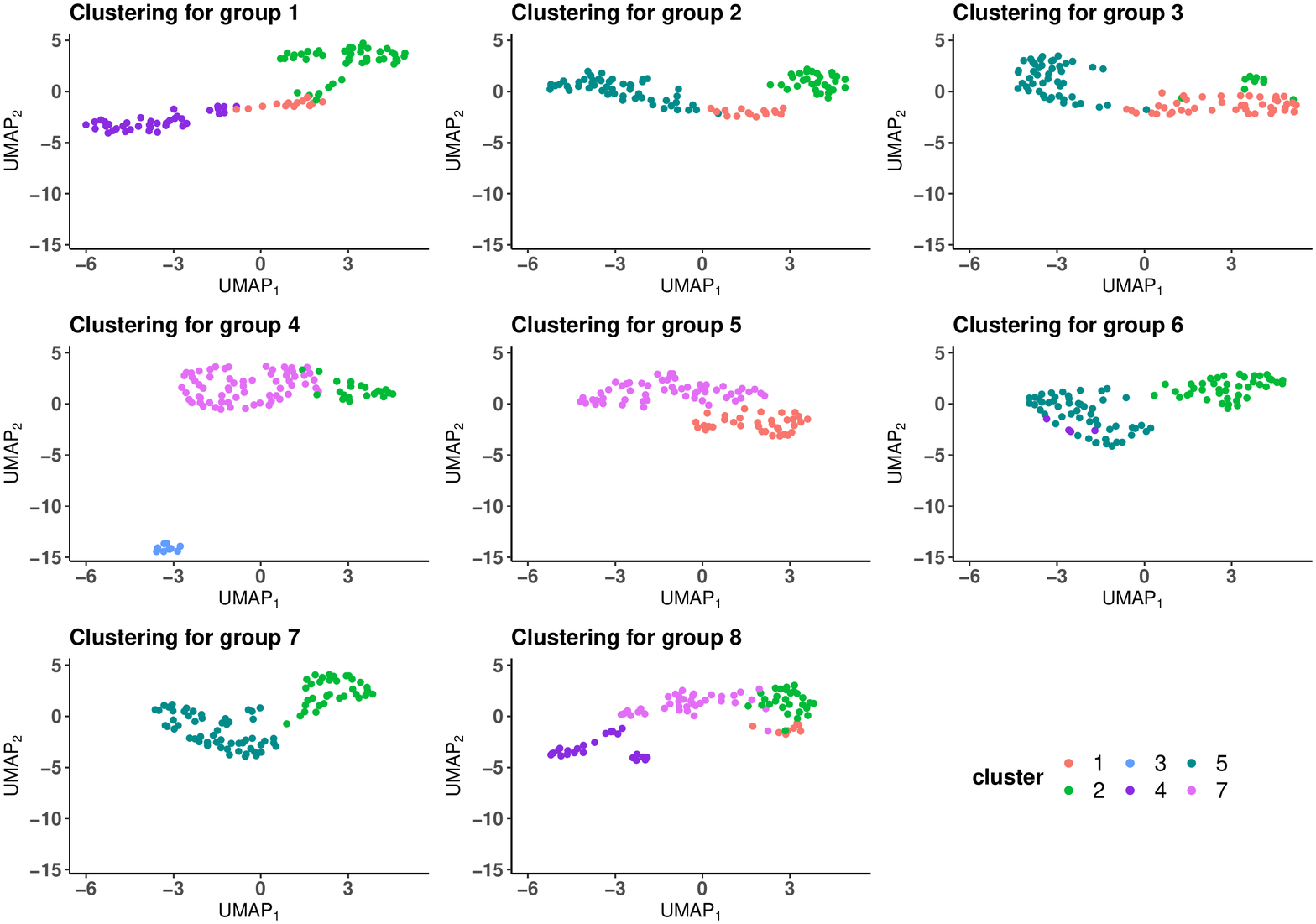}
  \caption{Clustering plot for different groups by GDP.}
  \label{fig:Real data cluster GDP}
\end{subfigure}%
\par\bigskip
\begin{subfigure}{.8\textwidth}
  \centering
  \includegraphics[width=\linewidth]{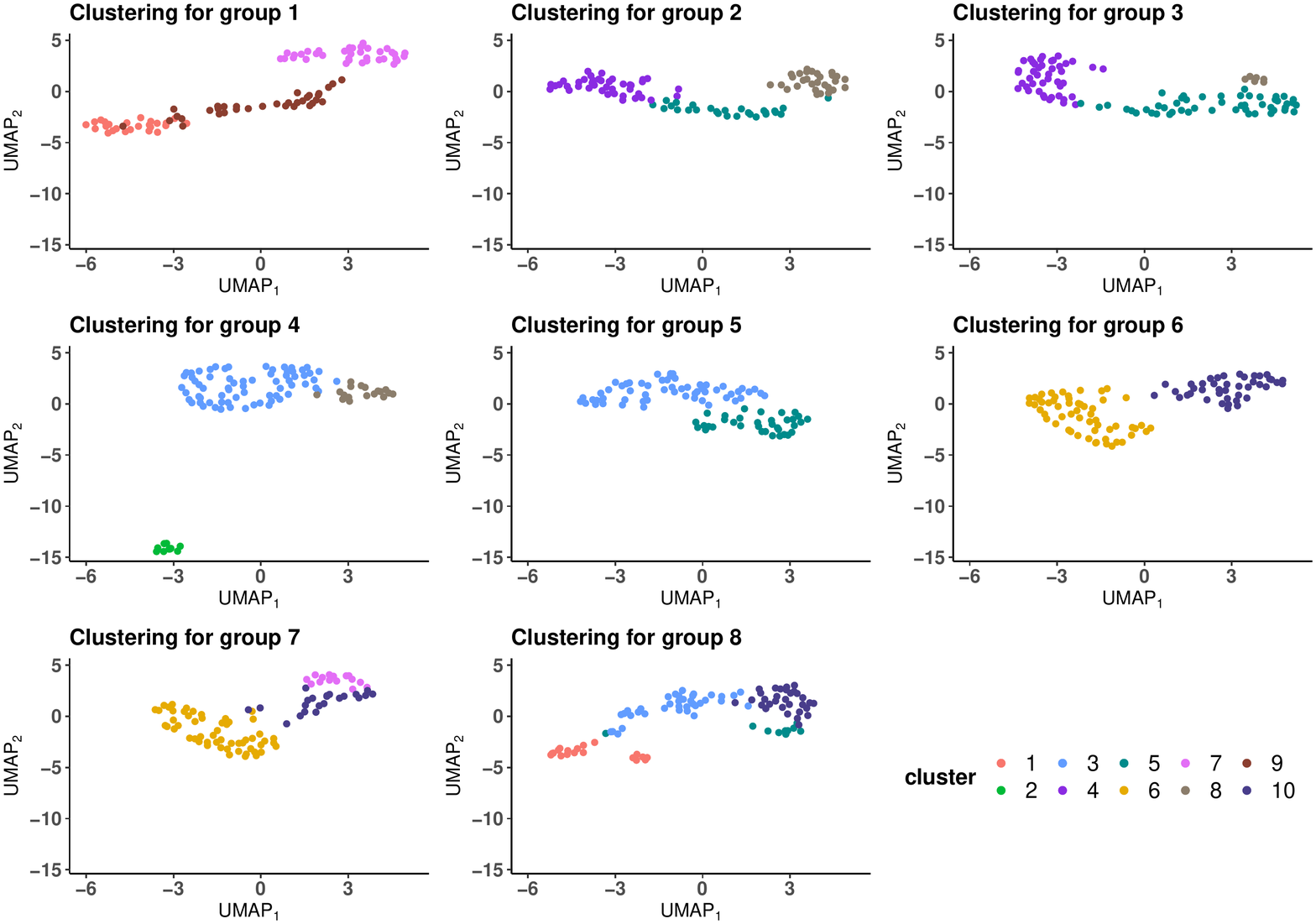}
  \caption{Clustering plot for different groups by HDP.}
  \label{fig:Real data cluster HDP}
\end{subfigure}
\caption{Clustering of the group-specific single-cell data whose dimensions are reduced to 2 by UMAP by (a) GDP and (b) HDP.}
\label{fig:Real data cluster}
\end{figure}

The estimated clusters from GDP and HDP are shown in Figures \ref{fig:Real data cluster GDP} and \ref{fig:Real data cluster HDP}, respectively. 
 As shown in Table \ref{supp-table::study design} in Section \ref{supp-sec::fmm_posterior} of the Supplementary Materials, group 1 is the wild-type group receiving the placebo and a normal diet.
 Each of group 2, 3, and 4 are obtained from group 1 by changing the three factors one at a time, and hence shares some similar clusters with group 1. Group 4 is similar to the baseline group 1 but with the Apc gene knocked out. The corresponding clustering plot (Figure \ref{fig:Real data cluster GDP}) of GDP indicates that the Apc knock-out group seems to exhibit more heterogeneity of cells (suggesting possibly new cellular subtypes) as compared to the wild-type group. Group 5 is the Apc knock-out group receiving a high-fat diet and the placebo. The clustering plot shows some resemblance with its parent groups (groups 2 and 4) but with the absence of some parental clusters. Groups 6 and 7 show similar clustering patterns, indicating possibly similar impact of changing the corresponding factors from their parent groups. Groups 7 and 8 correspond to the Apc knock-out group receiving the new treatment and fed with a normal and high-fat diet, respectively. It can be seen that the high-fat diet group appears to have greater molecular heterogeneity than the normal diet group. The Figure \ref{fig:Real data cluster HDP}, on the other hand, clearly shows that HDP fails to capture meaningful clusters across the non-exchangeable groups, i.e., some points that seemingly belong to the same cluster are assigned different labels across groups. To quantify the difference between GDP and HDP, we computed several internal clustering validation measures; 
 see \citealp{internal_clustering} for a review of several such measures. Table \ref{tab:internal clustering criterion} compares the Calinski-Harabasz, Davies–Bouldin, and Silhouette Index between GDP and HDP. Clearly, all of them indicate the superior clustering performance of GDP over HDP.

\begin{table}[!htb]
\centering
\begin{tabular}{ |c|c|c|c| }
 \hline
   & Calinski-Harabasz Index (CHI) & Davies–Bouldin Index (DBI) & Silhouette Index (SI) \\
  \hline
 GDP &  493.567 &  1.131 & 0.251\\
 HDP &  340.981 &  1.599 & 0.095\\
 \hline
\end{tabular}
\caption{Different measures of internal clustering for GDP and HDP. Higher values of CHI and SI indicates better clustering. Lower values of DBI indicate well separated clusters.}
\label{tab:internal clustering criterion}
\end{table}

\section{Discussion}\label{sec:discussion}
We have introduced the GDP as a graph-based stochastic process for modeling dependent random measures that are linked by a DAG. We have also introduced the corresponding infinite mixture model and presented how the GDP mixture model can be used for clustering grouped data with non-exchangeable groups. We provided different representations of the GDP including a novel hypergraph representation of the original process. The posterior inference was relatively straightforward. We illustrated our method using both simulations and an application to a real grouped scRNA-seq dataset. \par
There are a few possible future directions for this work. First, it may be possible to replace the DAG in our GDP with an undirected or chain graph. The challenge is to define the joint distribution over a set of random measures given the graph where the convenient DAG factorization no longer applies. Second, it may also be possible to learn the DAG structure instead of assuming it is known, which may require independent realizations of the GDP. Third, it will be interesting to extend the nested DP and other Bayesian nonparametric priors to grouped data with non-exchangeable groups.

\begin{appendices}

\section{Proof of the Hypergraph Representation}
\label{appendix_GDP_example}
We prove Theorem \ref{main theorem} (the hypergraph representation of the proposed GDP) of the main manuscript in the case of our motivational problem where we have 8 groups. 
Note that the proof for any general DAG follows in a similar fashion by repeated application of the two lemmas in Section \ref{sec:hyperpriors} of the main manuscript and properties of gamma and Dirichlet distributions, which, however, requires more involved bookkeeping of the corresponding random distributions and hence is omitted. 
Our proof also illustrates how the random distribution of any particular node of the DAG is related to the root node through a number of hidden random measures, which shows the clustering property of our model. 
In our motivating example, 
each group corresponds to a combination of treatment, diet, and genotype, as summarized in Table \ref{supp-table::study design} in Section \ref{supp-sec::fmm_posterior} of the Supplementary Materials. The underlying DAG for the problem is given in Figure \ref{fig::DAG} of the main manuscript
where group 1 is the root node, groups 2-4 are the layer-1 nodes, groups 5-7 are the layer-2 nodes, and group 8 is the layer-3 node. For ease of notation, instead of using $G_1^{(0)}$ and $\alpha_1^{(0)}$ to denote the random measure and the concentration parameter of the root node, we use simply $G_1$ and $\alpha_1$ instead; similarly for all the other nodes. 
Using these simplified notations, Figures \ref{DAG for random measures} and \ref{DAG for concentration parameters} show the relationships among the group-specific random measures and concentration parameters according to Figure \ref{fig::DAG} of the main manuscript. 

\begin{figure}[http]
\centering
    \begin{subfigure}[b]{0.45\textwidth}
        \centering
            \begin{tikzpicture}[node distance={15mm}, thick, main/.style = {draw, circle, scale = 1, transform shape}] 
\node[main, rectangle] (0) {$G_0$}; 
\node[main] (1) [below of=0] {$G_1$}; 
\node[main] (3) [below of=1] {$G_3$};
\node[main] (2) [left of=3] {$G_2$}; 
\node[main] (4) [right of=3] {$G_4$}; 
\node[main] (5) [below of=2] {$G_5$}; 
\node[main] (6) [below of=3] {$G_6$};
\node[main] (7) [below of=4] {$G_7$};
\node[main] (8) [below of=6] {$G_8$};
\draw[->] (0) -- (1); 
\draw[->] (1) -- (2); 
\draw[->] (1) -- (3); 
\draw[->] (1) -- (4);
\draw[->] (2) -- (5); 
\draw[->] (2) -- (6);
\draw[->] (3) -- (6); 
\draw[->] (3) -- (7); 
\draw[->] (4) -- (5);
\draw[->] (4) -- (7);
\draw[->] (5) -- (8);
\draw[->] (6) -- (8);
\draw[->] (7) -- (8);
\end{tikzpicture} 
\caption{}
\label{DAG for random measures}
    \end{subfigure}
    \begin{subfigure}[b]{0.45\textwidth}
    \centering
    \begin{tikzpicture}[node distance={15mm}, thick, main/.style = {draw, circle, scale = 1, transform shape}] 
\node[main, rectangle] (0) {$\alpha_0$}; 
\node[main] (1) [below of=0] {$\alpha_1$}; 
\node[main] (3) [below of=1] {$\alpha_3$};
\node[main] (2) [left of=3] {$\alpha_2$}; 
\node[main] (4) [right of=3] {$\alpha_4$}; 
\node[main] (5) [below of=2] {$\alpha_5$}; 
\node[main] (6) [below of=3] {$\alpha_6$};
\node[main] (7) [below of=4] {$\alpha_7$};
\node[main] (8) [below of=6] {$\alpha_8$}; 
\draw[->] (0) -- (1); 
\draw[->] (1) -- (2); 
\draw[->] (1) -- (3); 
\draw[->] (1) -- (4);
\draw[->] (2) -- (5); 
\draw[->] (2) -- (6);
\draw[->] (3) -- (6); 
\draw[->] (3) -- (7); 
\draw[->] (4) -- (5);
\draw[->] (4) -- (7);
\draw[->] (5) -- (8);
\draw[->] (6) -- (8);
\draw[->] (7) -- (8);
\end{tikzpicture} 
\caption{}
\label{DAG for concentration parameters}
    \end{subfigure}
    \caption{The DAG of the (a) random measures $G_j$'s and  (b) concentration parameters $\alpha_j$'s.}
\end{figure}
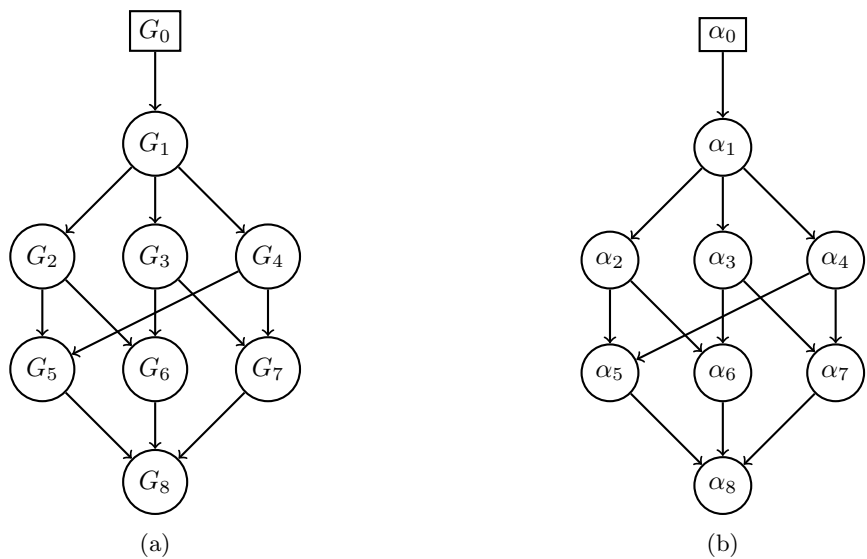

The proposed GDP mixture model for this problem is given hierarchically as,  
\begin{align}
\begin{aligned}
     \alpha_1 \mid \alpha_0 & \sim Gamma(\alpha_0, 1), & 
    G_1\mid G_0, \alpha_1 & \sim DP(\alpha_1, G_0),\\
    \alpha_2 \mid \alpha_1 & \sim Gamma(\alpha_1, 1), & G_2\mid G_1, \alpha_2 & \sim DP(\alpha_2, G_1),\\
    \alpha_3 \mid \alpha_1 & \sim Gamma(\alpha_1, 1), & G_3\mid G_1, \alpha_3 & \sim DP(\alpha_3, G_1),\\
    \alpha_4 \mid \alpha_1 & \sim Gamma(\alpha_1, 1), & G_4\mid G_1, \alpha_4 & \sim DP(\alpha_4, G_1),\\
    \alpha_5\mid \alpha_2, \alpha_4 & \sim Gamma(\alpha_2 + \alpha_4, 1), & G_5 \mid G_2, G_4, \alpha_5, \pi_1 & \sim DP(\alpha_5, \pi_1 G_2 + (1-\pi_1)G_4), \\
    & & \pi_1 \mid \alpha_2, \alpha_4 & \sim Beta(\alpha_2, \alpha_4),\\
    \alpha_6\mid \alpha_2, \alpha_3 & \sim Gamma(\alpha_2 + \alpha_3, 1), & G_6 \mid G_2, G_3, \alpha_6, \pi_2 & \sim DP(\alpha_6, \pi_2 G_2 + (1-\pi_2)G_3), \\
     & & \pi_2 \mid \alpha_2, \alpha_2 & \sim Beta(\alpha_2, \alpha_3),\\
    \alpha_7\mid \alpha_3, \alpha_4 & \sim Gamma(\alpha_3 + \alpha_4, 1), & G_7 \mid G_3, G_4, \alpha_7, \pi_3 & \sim DP(\alpha_7, \pi_3 G_3 + (1-\pi_3)G_4), \\
     & & \pi_3 \mid \alpha_3, \alpha_4& \sim Beta(\alpha_3, \alpha_4),\\
    \alpha_8\mid \alpha_5, \alpha_6, \alpha_7 & \sim Gamma(\alpha_5 + \alpha_6 + \alpha_7, 1), & G_8 \mid G_5, G_6, G_7, \alpha_8, \boldsymbol{\gamma} & \sim DP(\alpha_8, \gamma_1 G_5 + \gamma_2 G_6 + \gamma_3 G_7), \\
     & & \boldsymbol{\gamma} = (\gamma_1, \gamma_2, \gamma_3) \mid \alpha_5, \alpha_6, \alpha_7 & \sim Dir(\alpha_5, \alpha_6, \alpha_7),\\
    \theta_{ji} \mid G_j &\overset{ind}{\sim} G_j,\\
    x_{ji} \mid \theta_{ji}, G_j, G_0 &\overset{ind}{\sim} F(\theta_{ji}), &  & i=1, \dots, n_j, \hspace{0.2cm} j = 1, \dots, 8.
    \end{aligned}
\end{align} 
Now, from \cref{main theorem}, we have the following hypergraph representation, which we are going to prove, 
\begin{align}
\begin{aligned}
    \alpha_1 \mid \alpha_0 & \sim Gamma(\alpha_0, 1), & 
    G_1\mid G_0, \alpha_1 & \sim DP(\alpha_1, G_0),\\
    \alpha_2 \mid \alpha_1 & \sim Gamma(\alpha_1, 1), & G_2\mid G_1, \alpha_2 & \sim DP(\alpha_2, G_1),\\
    \alpha_3 \mid \alpha_1 & \sim Gamma(\alpha_1, 1), & G_3\mid G_1, \alpha_3 & \sim DP(\alpha_3, G_1),\\ 
    \alpha_4 \mid \alpha_1 & \sim Gamma(\alpha_1, 1), & G_4\mid G_1, \alpha_4 & \sim DP(\alpha_4, G_1),\\
    \alpha_5\mid \alpha_2, \alpha_4 & \sim Gamma(\alpha_2 + \alpha_4, 1), & G_5 \mid H_1, \alpha_5  & \sim DP(\alpha_5, H_1), \\
    & & H_1 \mid \alpha_2, \alpha_4& \sim DP(\alpha_2 + \alpha_4, G_1),\\ 
    \alpha_6\mid \alpha_2, \alpha_3 & \sim Gamma(\alpha_2 + \alpha_3, 1), & G_6 \mid H_2, \alpha_6  & \sim DP(\alpha_6, H_2), \\
    & & H_2 \mid \alpha_2, \alpha_3 & \sim DP(\alpha_2 + \alpha_3, G_1),\\ 
    \alpha_7\mid \alpha_3, \alpha_4 & \sim Gamma(\alpha_3 + \alpha_4, 1), & G_7 \mid H_3, \alpha_7  & \sim DP(\alpha_7, H_3), \\
     & & H_3 \mid \alpha_3, \alpha_4 & \sim DP(\alpha_3 + \alpha_4, G_1),\\
    \alpha_8\mid \alpha_5, \alpha_6, \alpha_7 & \sim Gamma(\alpha_5 + \alpha_6 + \alpha_7, 1), & G_8 \mid H_4, \alpha_8 & \sim DP(\alpha_8, H_4), \\
     & & H_4 \mid \alpha_5, \alpha_6, \alpha_7,  H^* & \sim DP(\alpha_5 + \alpha_6 + \alpha_7, H^*),\\
     & & H^* \mid \alpha_2, \alpha_3, \alpha_4,  G_1 & \sim DP(2(\alpha_2 + \alpha_3 + \alpha_4), G_1),\\
     \theta_{ji} \mid G_j &\overset{ind}{\sim} G_j,\\
     x_{ji} \mid \theta_{ji}, G_j, G_0 &\overset{ind}{\sim} F(\theta_{ji}), &  i=1, \dots, n_j, \hspace{0.2cm} j = 1, \dots, 8.
     \end{aligned}
    \label{Hypergraph Theorem proof}
\end{align}
\begin{proof}

Note that the random measures $G_2, G_3$, and $G_4$ are the layer-1 nodes. Their relationships to the root node $G_1$ are the same as those in an HDP. We shall consider the relationships  of the random measures of the layer-2 and layer-3 nodes (i.e., $G_5, G_6, G_7$, and $G_8$) to the root node.
Let $H_1 = \pi_1 G_2 + (1 - \pi_1) G_4$ where $G_2 \sim DP(\alpha_2, G_1)$ and $G_4  \sim DP(\alpha_4, G_1)$ independently.
Let $A_1, A_2, \dots, A_r$ be a finite measurable partition of the sample space $\Theta$. Then by the definition of DP, we have
\begin{align}
    \nonumber \left(G_2(A_1), G_2(A_2), \dots, G_2(A_r)\right) & \sim Dir\left(\alpha_2 G_1(A_1), \alpha_2 G_1(A_2), \dots, \alpha_2 G_1(A_r)\right),\\
    \nonumber  \left(G_4(A_1), G_4(A_2), \dots, G_4(A_r)\right) & \sim Dir\left(\alpha_4 G_1(A_1), \alpha_4 G_1(A_2), \dots, \alpha_4 G_1(A_r)\right), 
\end{align}
which are conditionally independent given $\alpha_2, \alpha_4$ and $G_1$.
As $\pi_1\sim Beta(\alpha_2, \alpha_4)$ 
independently of $G_2$ and $G_4$, using \cref{lemmaI}, we have that, given $\alpha_2, \alpha_4$ and $G_1$,
\begin{align}
    \nonumber \pi_1 \left(G_2(A_1), \dots, G_2(A_r)\right) + (1-\pi)\left(G_4(A_1), \dots, G_4(A_r)\right) &\sim Dir((\alpha_2 + \alpha_4) (G_1(A_1), \dots, G_1(A_r)))\\
    \nonumber \Rightarrow (H_1(A_1), \dots, H_1(A_r)) \mid \alpha_2, \alpha_4, G_1 & \sim Dir((\alpha_2 + \alpha_4)  (G_1(A_1), \dots, G_1(A_r)))\\
    \Rightarrow  H_1 \mid \alpha_2, \alpha_4, G_1 & \sim DP(\alpha_2 + \alpha_4, G_1)
\end{align}
Thus, we have 
\begin{align}
\label{hidden eq 1}
    \nonumber G_5 \mid H_1, \alpha_5 & \sim DP(\alpha_5, H_1),\\
    H_1 \mid \alpha_2, \alpha_4, G_1 & \sim DP(\alpha_2 + \alpha_4, G_1).
\end{align}
Similarly, the other layer-2 measures $G_6$ and $G_7$ have the following representations:
\begin{align}
\label{hidden eq 2}
    \nonumber G_6 \mid H_2, \alpha_6 & \sim DP(\alpha_6, H_2),\\
    H_2 \mid \alpha_2, \alpha_3, G_1 & \sim DP(\alpha_2 + \alpha_3, G_1),
 \end{align}   
 and,
\begin{align}
\label{hidden eq 3}
    \nonumber G_7 \mid H_3, \alpha_7 & \sim DP(\alpha_7, H_3),\\
    H_3 \mid \alpha_3, \alpha_4, G_1 & \sim DP(\alpha_3 + \alpha_4, G_1),
\end{align}
where $H_2 = \pi_2 G_2 + (1 - \pi_2) G_3$ and $H_3= \pi_3 G_3 + (1 - \pi_3) G_4$.

Let $H_4 = \gamma_1 G_5 + \gamma_2 G_6 + \gamma_3 G_7$ and $\boldsymbol{\gamma} = (\gamma_1, \gamma_2, \gamma_3) \sim Dir(\alpha_5, \alpha_6, \alpha_7)$. Since $G_5$, $G_6$, and $G_7$ are conditionally independent given $G_2$, $G_3$, and $G_4$, they are also independent given $H_1, H_2$, and $H_3$. Therefore, we have,
\begin{align*}
    G_5 \mid \alpha_5, H_1 & \sim DP(\alpha_5, H_1),\\
    G_6 \mid \alpha_6, H_2 & \sim DP(\alpha_6, H_2),\\
    G_7 \mid \alpha_7, H_3 & \sim DP(\alpha_7, H_3).
\end{align*}
For any finite measurable partition $A_1, A_2, \dots, A_r$ of $\Theta$, from \cref{lemmaII}, we have
\begin{flalign}
\label{sim eq1}
   \nonumber & \left(H_4(A_1), \dots, H_4(A_r)\right)  \mid \alpha_5, \alpha_6, \alpha_7, H_1, H_2, H_3 &&\\ 
   = \nonumber \hspace{0.3cm}  & \gamma_1 \left(G_5(A_1), \dots, G_5(A_r)\right) + \gamma_2 \left(G_6(A_1), \dots, G_6(A_r)\right)  + \gamma_3 \left(G_7(A_1), \dots, G_7(A_r)\right) &&\\
    \nonumber   & \phantom{\left(G_7(A_1), \dots,\right) } \sim Dir\left(\left(\alpha_5 H_1 + \alpha_6 H_2 + \alpha_7 H_3\right)(A_1), \dots
    , \left(\alpha_5 H_1 + \alpha_6 H_2 + \alpha_7 H_3\right)(A_r)\right) && \\
    \nonumber  & \phantom{\left(G_7(A_1), \dots,\right) } \equiv Dir\left(\alpha^* \left(\left(\frac{\alpha_5}{\alpha^*}H_1 + \frac{\alpha_6}{\alpha^*}H_2 + \frac{\alpha_7}{\alpha^*} H_3 \right)(A_1), \dots, \left(\frac{\alpha_5}{\alpha^*}H_1 + \frac{\alpha_6}{\alpha^*}H_2 + \frac{\alpha_7}{\alpha^*} H_3\right) (A_r)\right)   \right) &&\\
    \nonumber  & \phantom{\left(G_7(A_1), \dots,\right) } \equiv Dir\left(\alpha^* \left(H^*(A_1), \dots, H^* (A_r)\right)   \right) &&\\
    & \phantom{\left(G_7(A_1), \dots,\right) }\Rightarrow H_4 \mid \alpha^*, H^*  \sim DP(\alpha^*, H^*),
\end{flalign}
\noindent where $\alpha^* = \alpha_5 + \alpha_6 + \alpha_7$ and $H^* = \frac{\alpha_5}{\alpha^*}H_1 + \frac{\alpha_6}{\alpha^*}H_2 + \frac{\alpha_7}{\alpha^*} H_3 $. Note that 
$\alpha_5$, $\alpha_6$, and $\alpha_7$ are independent gamma random variables conditionally on $\alpha_2, \alpha_3, \alpha_4$ with shape parameters $\alpha_2 + \alpha_4, \alpha_2 + \alpha_3$, and $\alpha_3 + \alpha_4$, respectively. Thus,
\begin{equation}
\label{sim eq2}
    \left(\frac{\alpha_5}{\alpha^*}, \frac{\alpha_6}{\alpha^*}, \frac{\alpha_7}{\alpha^*}\right) \mid \alpha_2, \alpha_3, \alpha_4 \sim Dir(\alpha_2+ \alpha_4, \alpha_2 + \alpha_3, \alpha_3 + \alpha_4).
\end{equation}
Thus, given $G_1, G_2, G_3, G_4, \alpha_2, \alpha_3, \alpha_4$, and from \crefrange{hidden eq 1}{hidden eq 3}, and \cref{sim eq2}, using \cref{lemmaII} and using a similar measurable finite partition of $\Theta$ argument, we have,
\begin{equation}
    H^* \mid \alpha_2, \alpha_3, \alpha_4, G_1 \sim DP(2(\alpha_2 + \alpha_3 + \alpha_4), G_1),
\end{equation}
which completes the proof.
\end{proof}

\end{appendices}

\bigskip
\begin{center}
{\Large \bf Supplementary Materials for ``Graphical Dirichlet Process for Clustering Non-Exchangeable Grouped Data"}
\end{center}

\beginsupplement

\section{Proof of Lemma 1 and Lemma 2}
\label{supp-sec::proof_lemma}
\begin{lemma}[\citealp{Sethuraman}]
\label{supp-lemmaI}
Let $\boldsymbol{\alpha}_1 = (\alpha_{11}, \alpha_{12}, \dots, \alpha_{1k})$ and $\boldsymbol{\alpha}_2 = (\alpha_{21}, \alpha_{22}, \dots, \alpha_{2k})$ be $k$-dimensional vectors with $\alpha_{ij} > 0 \hspace{0.2cm} \forall \hspace{0.2cm} j = 1, 2, \dots ,k,\hspace{0.2cm} i = 1, 2$. Let $\boldsymbol{X}_1$ and $\boldsymbol{X}_2$ be independent $k$-dimensional random vectors distributed as Dirichlet distribution with parameters $\boldsymbol{\alpha}_1$ and $\boldsymbol{\alpha}_2$, respectively. Let $\alpha_{1 \cdot} = \sum_{j=1}^{k} \alpha_{1j}$ and $\alpha_{2\cdot} = \sum_{j=1}^{k} \alpha_{2j}$. Let $\pi$ be independent of $\boldsymbol{X}_1$ and $\boldsymbol{X}_2$ and have a beta distribution $Beta\left(\alpha_{1\cdot}, \alpha_{2\cdot}\right)$. Then the distribution of $\pi \boldsymbol{X}_1 + \left(1 -\pi\right)\boldsymbol{X}_2$ is the Dirichlet distribution with parameter $\boldsymbol{\alpha}_1 + \boldsymbol{\alpha}_2$.
\end{lemma} 
\begin{proof}
Let $T_i \overset{ind}{\sim} Gamma\left(\alpha_{1i}, \lambda\right), \hspace{0.2cm} i = 1, 2, \dots, k$ and $S_i \overset{ind}{\sim} Gamma\left(\alpha_{2i}, \lambda\right), \hspace{0.2cm} i = 1, 2, \dots, k$ independently of $T_i$, where $\lambda > 0$. Let $T = \sum_{i = 1}^{k} T_i$ and $S = \sum_{i = 1}^{k} S_i$. We know from the reproductive property of independent gamma distributions that $T \sim Gamma\left(\sum_{i = 1}^{k} \alpha_{1i}, \lambda\right) \equiv Gamma\left(\alpha_{1\cdot}, \lambda\right)$ and $S \sim Gamma\left(\sum_{i = 1}^{k} \alpha_{2i}, \lambda\right) \equiv Gamma\left(\alpha_{2\cdot}, \lambda\right)$ independently of $T$. Define
\begin{align*}
    \boldsymbol{X}_1  := \left(\frac{T_1}{T}, \frac{T_2}{T},\dots, \frac{T_k}{T}\right),  \hspace{0.4cm}
    \boldsymbol{X}_2  := \left(\frac{S_1}{S}, \frac{S_2}{S},\dots, \frac{S_k}{S}\right), \hspace{0.4cm} \text{and}\hspace{0.4cm} \pi := \frac{T}{T+S}. 
\end{align*}
\sloppy It is easy to see that
$\boldsymbol{X}_1 \sim Dir(\alpha_{11}, \dots, \alpha_{1k})$ is independent of $\boldsymbol{X}_2 \sim Dir(\alpha_{21}, \dots, \alpha_{2k})$, and that $\pi \sim Beta\left(\alpha_{1\cdot}, \alpha_{2\cdot}\right)$. We now need to show that $\pi$ as defined above is indeed independent of $\boldsymbol{X}_1$ and $\boldsymbol{X}_2$ as required by the lemma. For any fixed $\alpha_{11}, \dots, \alpha_{1k}$, we have that $\sum_{i=1}^{k}T_i$ is a complete and sufficient statistic for $\lambda$. Because $\boldsymbol{X}_1
\sim Dir(\alpha_{11}, \dots, \alpha_{1k})$ is ancillary for $\lambda$, by the Basu's theorem \citep{basu}, we have that $\boldsymbol{X}_1$ is independent of $\sum_{i=1}^{k}T_i = T$. 
Furthermore, due to the independence of $S_i$ and $T_i$, $i = 1, \dots, k$, $\boldsymbol{X}_1$ is independent of $S$, and,  therefore, $\boldsymbol{X}_1$  is independent of $\pi = \frac{T}{T+S}$. Similarly, $\boldsymbol{X}_2$ is also independent of $\pi$.
Then,
\begin{align*}
    \pi \boldsymbol{X}_1 + \left(1 - \pi\right) \boldsymbol{X}_2 & = \frac{T}{T+S} \left(\frac{T_1}{T}, \frac{T_2}{T},\dots, \frac{T_k}{T}\right) + \frac{S}{T+S} \left(\frac{S_1}{S}, \frac{S_2}{S},\dots, \frac{S_k}{S}\right)\\
    & = \left(\frac{T_1 + S_1}{T + S}, \frac{T_2 + S_2}{T + S},\dots, \frac{T_k + S_k}{T + S}\right) \sim Dir\left(\boldsymbol{\alpha}_1 + \boldsymbol{\alpha}_2\right),
\end{align*}
because $T_i + S_i \overset{ind}{\sim} Gamma(\alpha_{1i} + \alpha_{2i}, \lambda) \hspace{0.2cm} i = 1, 2, \dots, k$ and $T + S \sim Gamma(\alpha_{1\cdot} + \alpha_{2\cdot}, \lambda) $.
\end{proof}

\begin{lemma}
\label{supp-lemmaII}
\sloppy Let $\boldsymbol{\alpha}_1, \boldsymbol{\alpha}_2, \dots, \boldsymbol{\alpha}_L$ be $k$-dimensional vectors where $\boldsymbol{\alpha}_i = (\alpha_{i1},\dots, \alpha_{ik})$ with $\alpha_{ij} > 0 \hspace{0.2cm} \forall \hspace{0.2cm} j = 1, 2, \dots ,k$,  $i = 1,2, \dots, L$. Let $\boldsymbol{X}_1, \boldsymbol{X}_2, \dots, \boldsymbol{X}_L$ be independent $k$-dimensional random vectors distributed as Dirichlet distribution with parameters $\boldsymbol{\alpha}_1, \boldsymbol{\alpha}_2,\dots,  \boldsymbol{\alpha}_L$, respectively. Let $\alpha_{i \cdot} = \sum_{j=1}^{k} \alpha_{ij},\ \ i = 1,2,\dots, L$. Let $\boldsymbol{\pi} = \left(\pi_1, \pi_2, \dots, \pi_L\right)$ be independent of $\boldsymbol{X}_1, \boldsymbol{X}_2,\dots, \boldsymbol{X}_L$ and have a Dirichlet distribution $Dir\left(\alpha_{1\cdot}, \alpha_{2\cdot},\dots,  \alpha_{L\cdot}\right)$. Then the distribution of $\sum_{i=1}^{L}\pi_i \boldsymbol{X}_i$ is the Dirichlet distribution with parameter $\sum_{i=1}^{L}\boldsymbol{\alpha}_i$. 
\end{lemma}
\begin{proof}
The proof is similar to that of \cref{supp-lemmaI}. By noting that 
\begin{equation*}
    \boldsymbol{\pi} \sim Dir\left(\alpha_{1\cdot}, \alpha_{2\cdot},\dots, \alpha_{L\cdot}\right) \overset{d}{=} \left(\frac{\gamma_1}{\gamma}, \frac{\gamma_2}{\gamma},\dots, \frac{\gamma_L}{\gamma}\right),
\end{equation*}
where $\gamma_i \overset{ind}{\sim} Gamma(\alpha_{i\cdot}, \lambda), \hspace{0.2cm} i= 1, 2, \dots, L$ and $\gamma = \sum_{i = 1}^{L} \gamma_i \sim Gamma\left(\sum_{i=1}^{L}\alpha_{i\cdot}, \lambda\right)$. The remaining proof follows from standard properties of Dirichlet distributions and mimics the proof of \cref{supp-lemmaI}.
\end{proof}

\section{Proof of the Infinite Limit of Finite Mixture Model}
\label{supp-sec::proof_infinite_limit}
The GDP mixture model can  be derived as the infinite limit of a finite mixture model. Let us denote the observations and the mixture component indicator from node $j$ in layer $k$ of DAG $D$ by $x_{ji}^{(k)}$ and $z_{ji}^{(k)}$, respectively. Let $\boldsymbol{\beta}_1^{(0)}$ be the vector of mixing weights for the root node. Denoting by $\boldsymbol{\beta}_j^{(k)}$ the mixing weights of node $j$ in layer $k$ and by $\nu_j^{(k,m)}$ the corresponding mixing weights for the hidden layer $m$, with $m=2,\dots, k$, we have
\begin{align}
\begin{aligned}
    \boldsymbol{\beta}_1^{(0)} \mid \alpha_1^{(0)} & \sim Dir\left(\alpha_1^{(0)}/L, \dots, \alpha_1^{(0)}/L\right),\\
    \boldsymbol{\nu}_j^{(k,2)} \mid \{\alpha_l^{(1)} : l \in an^{(k, k-1)}(j)\}, \boldsymbol{\beta}_1^{(0)} &\sim Dir\left(\sum_{l\in an^{(k, k-1)}(j)} \alpha_l^{(1)} \left(\beta_{11}^{(0)}, \dots, \beta_{1L}^{(0)}\right)\right),\\
     & \vdots\\
    \boldsymbol{\nu}_j^{(k,k)} \mid \{\alpha_l^{(k,k-1)} : l \in an^{(k, 1)}(j)\}, \boldsymbol{\nu}_j^{(k,k-1)} & \sim Dir\left(\sum_{l\in an^{(k, 1)}(j)} \alpha_l^{(k-1)} \left(\nu_{j1}^{(k,k-1)}, \dots, \nu_{jL}^{(k,k-1)}\right)\right),\\
    \boldsymbol{\beta}_j^{(k)} \mid \alpha_j^{(k)}, \boldsymbol{\nu}_j^{(k,k)} & \sim Dir\left(\alpha_j^{(k)} \left(\nu_{j1}^{(k,k)}, \dots, \nu_{jL}^{(k,k)}\right)\right),\\
     \phi_l \mid G_0 & \sim G_0, \\
     z_{ji}^{(k)} \mid \boldsymbol{\beta}_j^{(k)} &\sim \boldsymbol{\beta}_j^{(k)},\\
      x_{ji}^{(k)} \mid  z_{ji}^{(k)}, \left(\phi_l\right)_{l=1}^{L} & \sim F\left(\phi_{z_{ji}^{(k)}}\right).
      \end{aligned}
\end{align}
\begin{proof}
Consider the random probability measure 
\begin{equation*}
    G_1^{(0)\, L} = \sum_{l=1}^{L}\beta_{1l}^{(0)}\delta_{\phi_l}. 
\end{equation*}
\citealp{ishwaran_zarepour} shows that for every measurable function $g$, integrable with respect to $G_0$, we have, given $\alpha_1^{(0)}$, as $L\rightarrow \infty$
\begin{equation*}
    \int g(\theta)dG_1^{(0), L}(\theta) \overset{\mathcal{D}}{\rightarrow}
    \int g(\theta)dG_1^{(0)}(\theta). 
\end{equation*}
Further, consider 
\begin{align*}
    & G_j^{(k) \, L} = \sum_{l=1}^{L}\beta_{jl}^{(k)}\delta_{\phi_{l}},\\
    & H_j^{(k,m)\, L} = \sum_{l=1}^{L}\nu_{jl}^{(k,m)}\delta_{\phi_{l}}, \hspace{0.4cm} m = 2, \dots, k.
\end{align*}
Let $(A_1, \dots, A_r)$ be a measurable partition of the sample space $\Theta$. Let $K_t = \{l=1,  \dots, L : \phi_l \in A_t\}, \ t = 1, \dots, r$, where $r \leq L$. Assuming that $G_0$ is non-atomic, the $\phi_l$'s are distinct with probability one, implying that any partition of $\{1, \dots, L\}$ corresponds to some partition of $\Theta$. Thus, as $\boldsymbol{\beta}_j^{(k)} \mid \alpha_j^{(k)}, \boldsymbol{\nu}_j^{(k,k)} \sim Dir\left(\alpha_j^{(k)} \left(\nu_{j1}^{(k,k)}, \dots, \nu_{jL}^{(k,k)}\right)\right)$, from the properties of Dirichlet distribution, we have, 
\begin{equation*}
    \left(G_j^{(k)\, L}(A_1), \dots, G_j^{(k)\, L}(A_r)\right) = \left(\sum_{l\in K_1}\beta_{jl}^{(k)},\dots,\sum_{l\in K_r}\beta_{jl}^{(k)} \right) \sim Dir\left(\alpha_j^{(k)}\sum_{l\in K_1}\nu_{jl}^{(k,k)} , \dots, \alpha_j^{(k)} \sum_{l\in K_r}\nu_{jl}^{(k,k)}\right).
\end{equation*}
Thus, 
\begin{align*}
    & G_j^{(k)\, L} \mid  \alpha_j^{(k)}, H_j^{(k,k)\, L} \sim DP\left(\alpha_j^{(k)},  H_j^{(k,k)\, L}\right).
\end{align*}
Similarly,
    \begin{align*}
    & H_j^{(k,k)\, L} \mid \{\alpha_l^{(k-1)} : l \in an^{(k, 1)}(j)\}, H_j^{(k,k-1)\, L} \sim DP\left(\sum_{l\in an^{(k, 1)}(j)} \alpha_l^{(k-1)}, H_j^{(k,k-1)\, L}\right),\\
    & H_j^{(k,k-1)\, L} \mid \{\alpha_l^{(k-2)} : l \in an^{(k, 2)}(j)\}, H_j^{(k,k-2)\, L} \sim DP\left(\sum_{l\in an^{(k, 2)}(j)} \alpha_l^{(k-2)}, H_j^{(k,k-2)\, L}\right),\\
    & \vdots \\
    & H_j^{(k,2)\, L} \mid \{\alpha_l^{(1)} : l \in  an^{(k, k-1)}(j)\}, G_1^{(0)\, L} \sim DP\left(\sum_{l\in an^{(k, k-1)}(j)} \alpha_l^{(1)}, G_1^{(0)\, L}\right).
\end{align*}
By letting $L\rightarrow \infty$, the marginal distribution that this finite mixture model induces on the observations, $\boldsymbol{x}_j^{(k)} = (x_{j1}^{(k)}, x_{j2}^{(k)},\dots )$, approaches the proposed GDP mixture model. 
\end{proof}

\section{Finite Mixture Model Approximation and Posterior Inference}
\label{supp-sec::fmm_posterior}
The posterior inference of the proposed GDP mixture model is carried out using a blocked Gibbs sampler. For concreteness, we will present the finite mixture model approximation of the GDP for our motivating example and posterior inference based on this approximation.  
In our motivating application, we have 8 experimental groups. Each group corresponds to a combination of treatment, diet, and genotype; see Table \ref{supp-table::study design} where we use binary indicators to denote the genotype, the two levels of diet, and the two treatment regimes. The design of the experiments naturally introduces dependencies among the experimental groups, which are represented by the DAG in Figure \ref{supp-fig::DAG}, where group 1 is the root node, groups 2-4 are the layer-1 nodes, groups 5-7 are the layer-2 nodes, and group 8 is the layer-3 node. For ease of notation, instead of using $G_1^{(0)}$ and $\alpha_1^{(0)}$ to denote the random measure and the concentration parameter of the root node, we use simply $G_1$ and $\alpha_1$ instead; similarly for all the other nodes. 

\begin{figure}[htp]
\begin{tabular}{*{2}{>{\arraybackslash}b{\dimexpr0.5\linewidth-2\tabcolsep\relax}}}
\begin{center}

\begin{tikzpicture}[node distance={15mm}, thick, main/.style = {draw, circle}] 
\node[main] (1) {$ 1$}; 
\node[main] (3) [below of=1] {$ 3$};
\node[main] (2) [left of=3] {$ 2$}; 
\node[main] (4) [right of=3] {$ 4$}; 
\node[main] (5) [below of=2] {$ 5$}; 
\node[main] (6) [below of=3] {$ 6$};
\node[main] (7) [below of=4] {$ 7$};
\node[main] (8) [below of=6] {$ 8$}; 
\draw[->] (1) -- (2); 
\draw[->] (1) -- (3); 
\draw[->] (1) -- (4);
\draw[->] (2) -- (5); 
\draw[->] (2) -- (6);
\draw[->] (3) -- (6); 
\draw[->] (3) -- (7); 
\draw[->] (4) -- (5);
\draw[->] (4) -- (7);
\draw[->] (5) -- (8);
\draw[->] (6) -- (8);
\draw[->] (7) -- (8);
\end{tikzpicture} 
\caption{The DAG of experimental groups.}
\label{supp-fig::DAG}
\end{center}
    &
\renewcommand{\arraystretch}{1.3}
\begin{tabular}{ |c|c|c|c| } 
 \hline
 Group & Diet & Treatment & Genotype \\ 
 \hline
 $ 1$ & 0 & 0 & 0 \\ 
 $ 2$ & 1 & 0 & 0 \\ 
 $ 3$ & 0 & 1 & 0 \\
 $ 4$ & 0 & 0 & 1 \\
 $ 5$ & 1 & 0 & 1 \\ 
 $ 6$ & 1 & 1 & 0 \\ 
 $ 7$ & 0 & 1 & 1 \\ 
 $ 8$ & 1 & 1 & 1 \\ 
 \hline
\end{tabular}
\captionof{table}{Each experimental group corresponds to a combination of diet, treatment, and genotype. Diet = 1 corresponds to high-fat diet and 0 corresponds to normal diet, 
Treatment = 1 corresponds to AdipoRon  and 0 corresponds to no therapy, 
Genotype = 1 corresponds to Apc knock-out  and 0 corresponds to wild type.}
\label{supp-table::study design}
\end{tabular}
\end{figure}

Recall that from the main text, the finite truncation of the infinite mixture model representation is given by,
\begin{align}
\label{supp-eq::finit_mixture_model}
\begin{aligned}
    \boldsymbol{\beta}_1^{(0)} \mid \alpha_1^{(0)} & \sim Dir\left(\alpha_1^{(0)}/L, \dots, \alpha_1^{(0)}/L\right),\\
    \boldsymbol{\nu}_j^{(k,2)} \mid \{\alpha_l^{(1)} : l \in an^{(k, k-1)}(j)\}, \boldsymbol{\beta}_1^{(0)} &\sim Dir\left(\sum_{l\in an^{(k, k-1)}(j)} \alpha_l^{(1)} \left(\beta_{11}^{(0)}, \dots, \beta_{1L}^{(0)}\right)\right),\\
     & \vdots\\
    \boldsymbol{\nu}_j^{(k,k)} \mid \{\alpha_l^{(k,k-1)} : l \in an^{(k, 1)}(j)\}, \boldsymbol{\nu}_j^{(k-1)} & \sim Dir\left(\sum_{l\in an^{(k, 1)}(j)} \alpha_l^{(k-1)} \left(\nu_{j1}^{(k,k-1)}, \dots, \nu_{jL}^{(k,k-1)}\right)\right),\\
    \boldsymbol{\beta}_j^{(k)} \mid \alpha_j^{(k)}, \boldsymbol{\nu}_j^{(k,k)} & \sim Dir\left(\alpha_j^{(k)} \left(\nu_{j1}^{(k,k)}, \dots, \nu_{jL}^{(k,k)}\right)\right),\\
     \phi_l \mid G_0 & \sim G_0, \\
     z_{ji}^{(k)} \mid \boldsymbol{\beta}_j^{(k)} &\sim \boldsymbol{\beta}_j^{(k)},\\
      x_{ji}^{(k)} \mid  z_{ji}^{(k)}, \left(\phi_l\right)_{l=1}^{L} & \sim F\left(\phi_{z_{ji}^{(k)}}\right).
      \end{aligned}
\end{align} 
Using the simplified notations for the group-specific random measures and concentration parameter, from the finite truncation of the infinite mixture model representation from Eq. \eqref{supp-eq::finit_mixture_model}, we have, for this motivating problem,
\begin{alignat}{3}
\label{supp-eq::finite_mixture_model_example_appendix}
     \nonumber \alpha_1 \mid \alpha_0 & \sim Gamma(\alpha_0, 1), & 
    \boldsymbol{\beta}_1 \mid \alpha_1 & \sim Dir(\alpha_1/L, \dots, \alpha_1/L), \\
    \nonumber 
    \alpha_2 \mid \alpha_1 & \sim Gamma(\alpha_1, 1), & \boldsymbol{\beta}_2 \mid \alpha_2, \boldsymbol{\beta}_1 & \sim Dir(\alpha_2\beta_{11}, \dots, \alpha_2\beta_{1L}), \\
    \nonumber
    \alpha_3 \mid \alpha_1 & \sim Gamma(\alpha_1, 1), & \boldsymbol{\beta}_3 \mid \alpha_3, \boldsymbol{\beta}_1 & \sim Dir(\alpha_3\beta_{11}, \dots, \alpha_3\beta_{1L}), \\
    \nonumber 
    \alpha_4 \mid \alpha_1 & \sim Gamma(\alpha_1, 1), & \boldsymbol{\beta}_4 \mid \alpha_4, \boldsymbol{\beta}_1 & \sim Dir(\alpha_4\beta_{11}, \dots, \alpha_4\beta_{1L}), \\
    \nonumber 
    \alpha_5\mid \alpha_2, \alpha_4 & \sim Gamma(\alpha_2 + \alpha_4, 1), & \boldsymbol{\beta}_5 \mid \alpha_5, \boldsymbol{\nu}_1 & \sim Dir(\alpha_5\nu_{11}, \dots, \alpha_5\nu_{1L}), \\
    \nonumber & & \boldsymbol{\nu}_1 \mid \alpha_2, \alpha_4, \boldsymbol{\beta}_1 & \sim Dir\left((\alpha_2 + \alpha_4) (\beta_{11},\dots, \beta_{1L})\right), \\
    \nonumber 
    \alpha_6\mid \alpha_2, \alpha_3 & \sim Gamma(\alpha_2 + \alpha_3, 1), & \boldsymbol{\beta}_6 \mid \alpha_6, \boldsymbol{\nu}_2 & \sim Dir(\alpha_6\nu_{21}, \dots, \alpha_6\nu_{2L}), \\
    \nonumber & & \boldsymbol{\nu}_2 \mid \alpha_2, \alpha_3, \boldsymbol{\beta}_1 & \sim Dir\left((\alpha_2 + \alpha_3)  (\beta_{11},\dots, \beta_{1L})\right),\\
    \nonumber 
    \alpha_7\mid \alpha_3, \alpha_4 & \sim Gamma(\alpha_3 + \alpha_4, 1), & \boldsymbol{\beta}_7 \mid \alpha_7, \boldsymbol{\nu}_3 & \sim Dir(\alpha_7\nu_{31}, \dots, \alpha_7\nu_{3L}), \\
    \nonumber & & \boldsymbol{\nu}_3 \mid \alpha_3, \alpha_4, \boldsymbol{\beta}_1 & \sim Dir\left((\alpha_3 + \alpha_4)  (\beta_{11},\dots, \beta_{1L})\right),\\
    \nonumber 
    \alpha_8\mid \alpha_5, \alpha_6, \alpha_7 & \sim Gamma(\alpha_5 + \alpha_6 + \alpha_7, 1), & \boldsymbol{\beta}_8 \mid \boldsymbol{\nu}_4, \alpha_8,  & \sim Dir(\alpha_8\nu_{41}, \dots, \alpha_8\nu_{4L}), \\
    \nonumber & & \boldsymbol{\nu}_4 \mid \alpha_5, \alpha_6, \alpha_7,  \boldsymbol{\eta} & \sim Dir\left((\alpha_5 + \alpha_6 + \alpha_7)(\eta_1, \dots, \eta_L)\right),\\
    \nonumber & & 
    \boldsymbol{\eta} \mid \alpha_2, \alpha_3, \alpha_4,  \boldsymbol{\beta}_1 & \sim Dir\left(2(\alpha_2 + \alpha_3 + \alpha_4) (\beta_{11},\dots, \beta_{1L})\right),\\
    \nonumber 
     z_{ji} \mid \boldsymbol{\beta}_j &\overset{ind}{\sim} Cat(1:L, \boldsymbol{\beta}_j),\\
     x_{ji} \mid z_{ji}, (\phi_l)_{l=1}^{L}  &\overset{ind}{\sim} F(\phi_{z_{ji}}),  \hspace{0.2cm}   & & i=1, \dots, n_j, \hspace{0.2cm} j = 1, \dots, 8.
\end{alignat}
With the above distributional structure, Gibbs sampling is straightforward. We use $\pi(.)$ and $\pi( .\mid -)$ to denote the prior distribution and the conditional distribution, respectively, of the parameter specified in the argument. The full conditional distribution for the atoms is given by,
\begin{equation}
    \pi(\{\phi_l\}_{l=1}^{L}\mid -) \propto \prod_{l=1}^{L}\left[\left\{\prod_{j=1}^{8}\prod_{i=1}^{n_j} F(x_{ji}\mid \phi_l)^{\mathds{1}(z_{ji} = l)}\right\} \pi( \phi_l)\right].
\end{equation}
The full conditional distributions for the latent cluster labels are given by,
\begin{equation}
    P(z_{ji} = l \mid -) \propto \beta_{jl} F(x_{ji}\mid \phi_l), \hspace{0.3cm} l = 1, \dots, L, \ \  i = 1, \dots, n_j \ \ j = 1,\dots, 8.
\end{equation}
The full conditional distribution for the stick-breaking weights is given by,
\begin{multline}
        \pi(\boldsymbol{\beta}_1\mid -) \propto \frac{\prod_{l=1}^{L} \beta_{1l}^{m_{1l} +\frac{\alpha_1}{L}}\left\{\beta_{2l}^{\alpha_2} \beta_{3l}^{\alpha_3}\beta_{4l}^{\alpha_4} \nu_{1l}^{\alpha_2 + \alpha_4}\nu_{2l}^{\alpha_2 + \alpha_3}\nu_{3l}^{\alpha_3 + \alpha_4} \eta_l^{2(\alpha_2 + \alpha_3 + \alpha_4)}\right\}^{\beta_{1l}}}{\prod_{l=1}^{L}\left\{\Gamma\left((\alpha_2 + \alpha_4 )\beta_{1l}\right) \Gamma\left((\alpha_2 + \alpha_3 )\beta_{1l}\right) \Gamma\left( (\alpha_3 + \alpha_4) \beta_{1l}\right) \Gamma\left(2(\alpha_2 + \alpha_3 + \alpha_4) \beta_{1l}\right)\right\}}\\
    \times \frac{1}{\prod_{l=1}^{L}\left\{ \Gamma\left(\alpha_2\beta_{1l}\right) \Gamma\left(\alpha_3\beta_{1l}\right) \Gamma\left(\alpha_4\beta_{1l}\right)\right\}},
\end{multline}
where $m_{1l} = \sum_{i =1}^{n_1} \mathds{1}(z_{1i} = l)$, $l = 1, \dots, L$.
The full conditionals for $\boldsymbol{\beta}_j, \ \ j = 2, \dots, 8$, are in closed form,
\begin{equation}
\label{beta1 eqn}
    \pi(\boldsymbol{\beta}_j\mid -) \sim Dir(\boldsymbol{m}_j + \alpha_j \boldsymbol{\beta}_1), \ \ \text{where $\boldsymbol{m}_j = (m_{j1}, \dots, m_{jL})$ and $m_{jl} = \sum_{i =1}^{n_j} \mathds{1}(z_{ji} = l)$, $l = 1, \dots, L$}.
\end{equation}
By letting $\boldsymbol{B}(\boldsymbol{a})$ to denote the multivariate beta function, i.e., for a $L$-dimensional vector $\boldsymbol{a} = (a_1, \dots, a_L)$ with $a_i > 0$, we have,
\begin{equation*}
    \boldsymbol{B}(\boldsymbol{a}) = \frac{\prod_{l=1}^{L}\Gamma(a_l)}{ \Gamma(\sum_{l=1}^{L} a_l)},
\end{equation*}
where $\Gamma(\cdot)$ is the gamma function. 
Then the full-conditional distribution of the hidden weights are given by,
\begin{equation}
\label{nu1 eqn}
    \pi(\boldsymbol{\nu}_1\mid -) \propto \frac{1}{\boldsymbol{B}(\alpha_5 \boldsymbol{\nu}_1)} \prod_{l=1}^{L} \left\{\beta_{5l}^{\alpha_5 \nu_{1l}}\nu_{1l}^{(\alpha_2 + \alpha_4)\beta_{1l} -1}\right\},
\end{equation}
\begin{equation}
\label{nu2 eqn}
    \pi(\boldsymbol{\nu}_2\mid -) \propto \frac{1}{\boldsymbol{B}(\alpha_6 \boldsymbol{\nu}_2)} \prod_{l=1}^{L} \left\{\beta_{6l}^{\alpha_6 \nu_{2l}}\nu_{2l}^{(\alpha_2 + \alpha_3)\beta_{1l} -1}\right\},
\end{equation}
\begin{equation}
\label{nu3 eqn}
    \pi(\boldsymbol{\nu}_3\mid -) \propto \frac{1}{\boldsymbol{B}(\alpha_7 \boldsymbol{\nu}_3)} \prod_{l=1}^{L} \left\{\beta_{7l}^{\alpha_7 \nu_{3l}}\nu_{3l}^{(\alpha_3 + \alpha_4)\beta_{1l} -1}\right\},
\end{equation}
\begin{equation}
\label{nu4 eqn}
    \pi(\boldsymbol{\nu}_4\mid -) \propto \frac{1}{\boldsymbol{B}(\alpha_8 \boldsymbol{\nu}_4)} \prod_{l=1}^{L} \left\{\beta_{4l}^{\alpha_8 \nu_{4l}}\nu_{4l}^{(\alpha_5 + \alpha_6 + \alpha_7)\eta_{l} -1}\right\},
\end{equation}
\begin{equation}
\label{eta eqn}
    \pi(\boldsymbol{\eta} \mid -) \propto \frac{1}{\boldsymbol{B}((\alpha_5 + \alpha_6 + \alpha_7) \boldsymbol{\eta})} \prod_{l=1}^{L} \left\{\eta_l^{2(\alpha_2 + \alpha_3 + \alpha_4)\beta_{1l} - 1}\nu_{4l}^{(\alpha_5 + \alpha_6 + \alpha_7)\eta_{l}}\right\}.
\end{equation}
The full conditionals for the concentration parameters are given by,
\begin{equation}
    \pi(\alpha_1 \mid -) \propto \frac{e^{-\alpha_1}\alpha_1^{\alpha_0 - 1} \alpha_2^{\alpha_1} \alpha_3^{\alpha_1}\alpha_4^{\alpha_1}}{\{\Gamma(\alpha)\}^3 \boldsymbol{B}((\alpha_1/L, \dots, \alpha_1/L)}\prod_{l=1}^{L}\beta_{1l}^{\frac{\alpha_1}{L}}
\end{equation}
\begin{equation}
    \pi(\alpha_2 \mid -) \propto \frac{e^{-\alpha_2} \alpha_2^{\alpha_1 -1} \alpha_5^{\alpha_2}\alpha_6^{\alpha_2} \left[\prod_{l=1}^{L}\left\{\beta_{2l}^{\beta_{1l} } \nu_{1l}^{\beta_{1l} }\nu_{2l}^{\beta_{1l}} \eta_{l}^{2 \beta_{1l}}\right\}^{\alpha_2}\right] \Gamma(\alpha_2) \Gamma(2(\alpha_2 + \alpha_3 + \alpha_4))}{\prod_{l=1}^{L}\left\{\Gamma(\alpha_2\beta_{1l}) \Gamma((\alpha_2 + \alpha_4)\beta_{1l}) \Gamma((\alpha_2 + \alpha_3)\beta_{1l}) \Gamma(2(\alpha_2 + \alpha_3 + \alpha_4)\beta_{1l})\right\}},
\end{equation}
\begin{equation}
    \pi(\alpha_3 \mid -) \propto \frac{e^{-\alpha_3} \alpha_3^{\alpha_1 -1} \alpha_6^{\alpha_3}\alpha_7^{\alpha_3} \left[\prod_{l=1}^{L}\left\{\beta_{3l}^{\beta_{1l} } \nu_{2l}^{\beta_{1l} }\nu_{3l}^{\beta_{1l}} \eta_{l}^{2 \beta_{1l}}\right\}^{\alpha_3}\right] \Gamma(\alpha_3) \Gamma(2(\alpha_2 + \alpha_3 + \alpha_4))}{\prod_{l=1}^{L}\left\{\Gamma(\alpha_3\beta_{1l}) \Gamma((\alpha_2 + \alpha_3)\beta_{1l}) \Gamma((\alpha_3 + \alpha_4)\beta_{1l}) \Gamma(2(\alpha_2 + \alpha_3 + \alpha_4)\beta_{1l})\right\}},
\end{equation}
\begin{equation}
    \pi(\alpha_4 \mid -) \propto \frac{e^{-\alpha_4} \alpha_4^{\alpha_1 -1} \alpha_5^{\alpha_4}\alpha_7^{\alpha_4} \left[\prod_{l=1}^{L}\left\{\beta_{4l}^{\beta_{1l} } \nu_{1l}^{\beta_{1l} }\nu_{3l}^{\beta_{1l}} \eta_{l}^{2 \beta_{1l}}\right\}^{\alpha_4}\right] \Gamma(\alpha_4) \Gamma(2(\alpha_2 + \alpha_3 + \alpha_4))}{\prod_{l=1}^{L}\left\{\Gamma(\alpha_4\beta_{1l}) \Gamma((\alpha_2 + \alpha_4)\beta_{1l}) \Gamma((\alpha_3 + \alpha_4)\beta_{1l}) \Gamma(2(\alpha_2 + \alpha_3 + \alpha_4)\beta_{1l})\right\}},
\end{equation}
\begin{equation}
    \pi(\alpha_5 \mid -) \propto \frac{e^{-\alpha_5} \alpha_5^{\alpha_2 + \alpha_4 -1} \alpha_8^{\alpha_5} \left[\prod_{l=1}^{L}\left\{\beta_{5l}^{\nu_{1l} } \nu_{4l}^{\eta_l}\right\}^{\alpha_5}\right] \Gamma(\alpha_5) }{\prod_{l=1}^{L}\left\{\Gamma(\alpha_5 \nu_{1l}) \Gamma((\alpha_5 + \alpha_6 + \alpha_7) \nu_l)\right\}},
\end{equation}
\begin{equation}
    \pi(\alpha_6 \mid -) \propto \frac{e^{-\alpha_6} \alpha_6^{\alpha_2 + \alpha_3 -1} \alpha_8^{\alpha_6} \left[\prod_{l=1}^{L}\left\{\beta_{6l}^{\nu_{2l} } \nu_{4l}^{\eta_l}\right\}^{\alpha_6}\right] \Gamma(\alpha_6) }{\prod_{l=1}^{L}\left\{\Gamma(\alpha_6 \nu_{2l}) \Gamma((\alpha_5 + \alpha_6 + \alpha_7) \nu_l)\right\}},
\end{equation}
\begin{equation}
    \pi(\alpha_7 \mid -) \propto \frac{e^{-\alpha_7} \alpha_7^{\alpha_3 + \alpha_4 -1} \alpha_8^{\alpha_7} \left[\prod_{l=1}^{L}\left\{\beta_{7l}^{\nu_{3l} } \nu_{4l}^{\eta_l}\right\}^{\alpha_7}\right] \Gamma(\alpha_7) }{\prod_{l=1}^{L}\left\{\Gamma(\alpha_7 \nu_{3l}) \Gamma((\alpha_5 + \alpha_6 + \alpha_7) \nu_l)\right\}},
\end{equation}
\begin{equation}
    \pi(\alpha_8 \mid -) \propto \frac{e^{-\alpha_8} \alpha_8^{\alpha_5 + \alpha_6 + \alpha_7 -1} \left[\prod_{l=1}^{L} \beta_{8l}^{\alpha_8\nu_{4l}}\right] \Gamma(\alpha_8) }{\prod_{l=1}^{L}\Gamma(\alpha_8 \nu_{4l})}.
\end{equation}
Note that the full conditionals of $\alpha_j, \ j = 1, \dots, 8$, $\boldsymbol{\beta}_1, \boldsymbol{\nu}_1, \boldsymbol{\nu}_2, \boldsymbol{\nu}_3, \boldsymbol{\nu}_4$, and $\boldsymbol{\eta}$ are not standard distributions that have direct samplers. We adopt a Metropolis-within-Gibbs strategy to sample from their corresponding full conditional distributions. Since $\alpha_j$'s are real-valued, sampling using a Metropolis step is straightforward. However, the main bottleneck in sampling are the weights $\boldsymbol{\beta}_1, \boldsymbol{\nu}_1, \boldsymbol{\nu}_2, \boldsymbol{\nu}_3, \boldsymbol{\nu}_4$ and $\boldsymbol{\eta}$, which have a complex structure on the simplex. To mitigate this problem, we use the SALTSampler \citep{SALTSampler} for which the implementation is publicly available as an R package.

\section{Simulation details}
\label{supp-sec:simulations}
Our simulations are designed to mimick the motivating application where we have 8 experimental groups. See Table \ref{supp-table::study design} for our experimental design represented in terms of binary indicators denoting the levels of diet, treatment, and genotype. The corresponding DAG is given in Figure \ref{supp-fig::DAG}.

For our simulation study, we generated data within each of the $8$ groups from a four-component mixture of bivariate Gaussian distributions with different covariance matrices for each group. Taking $\alpha_0=5$, we drew the concentration parameters for the different groups $\alpha_j$'s, the mixture model weights, $\boldsymbol{\beta}_j$'s, $\boldsymbol{\nu}_j$'s, and $\boldsymbol{\eta}_j$, and the true cluster indicators $z_{ji}$'s for each of the different groups using \eqref{supp-eq::finite_mixture_model_example_appendix}. Given the cluster indicators, the data were generated from the Gaussian distribution with the true cluster-specific means $\phi_l$'s given in Table \ref{supp-table:cluster means_simulation} and the group-specific covariance matrices given in Table \ref{supp-table:population covariance_simulation}. Note that within each group, the same covariance matrix was used for all clusters. 
\begin{table}[htp]
    \centering
    \begin{tabular}{c|c}
    \hline
    Cluster & Mean\\
    \hline
    1 &  (-2, -5) \\
    2 &  (0, 0)  \\
    3 &  (-3, 3) \\
    4 &  (3, -3)\\
    \hline
    \end{tabular}
    \caption{True cluster-specific  means.}
    \label{supp-table:cluster means_simulation}
\end{table}
\begin{table}[htp]
    \centering
    \begin{tabular}{c|c}
    \hline
    Group & Covariance\\
    \hline
    \\[-1em]
    1 &  $\begin{bmatrix}
0.8 & 0.3 \\
0.3 & 0.8
    \end{bmatrix}$
    \\
    \\[-1em]
    \hline
    \\[-1em]
    2 &  $\begin{bmatrix}
0.85 & 0.25 \\
0.25 & 0.85
    \end{bmatrix}$  \\
    \\[-1em]
    \hline
    \\[-1em]
    3 &  $\begin{bmatrix}
1 & 0.1 \\
0.1 & 1
    \end{bmatrix}$ \\
    \\[-1em]
    \hline
    \\[-1em]
    4 &  $\begin{bmatrix}
0.8 & -0.1 \\
-0.1 & 0.8
    \end{bmatrix}$\\
    \\[-1em]
    \hline
    \\[-1em]
    5 &  $\begin{bmatrix}
0.8 & -0.2 \\
-0.2 & 0.9
    \end{bmatrix}$\\
    \\[-1em]
    \hline
    \\[-1em]
    6 &  $\begin{bmatrix}
0.8 & 0 \\
0 & 0.8
    \end{bmatrix}$\\
    \\[-1em]
    \hline
    \\[-1em]
    7 &  $\begin{bmatrix}
0.75 & 0.25 \\
0.25 & 0.75
    \end{bmatrix}$\\
    \\[-1em]
    \hline
    \\[-1em]
    8 &  $\begin{bmatrix}
1.1 & 0.1 \\
0.1 & 1.1
    \end{bmatrix}$\\
    \hline
    \end{tabular}
    \caption{True covariance matrices for different groups.}
    \label{supp-table:population covariance_simulation}
\end{table}

In our Gibbs sampler, the truncation level of the finite mixture model was set to $L = 10$, and the base measure for GDP, $G_0$, was specified as the normal-inverse-Wishart distribution, $\mathcal{NIW}(\boldsymbol{0}, 0.01, \mathbb{I}_2, 2)$.
Upon the completion of the Gibbs sampler, the clusters were estimated by using the least squares criterion \citep{LeastSquares}, and they were compared with the true cluster labels for evaluation. We considered a variety of sample sizes as well as a case with very imbalanced design, which are summarized in Table \ref{supp-tab:ss}. In all cases, we ran 15,000 iterations of our Gibbs sampler and discarded the first 5,000 samples as burn-in.

\begin{table}[htp]
\centering
\begin{tabular}{ |c|c|c|c|c| }
 \hline
 Group & \multicolumn{4}{|c|}{Sample sizes} \\ 
 \cline{2-5}
  & small & moderate & large & unbalanced\\
 \hline
 $1$ & 40  & 80  & 150  & 350 \\ 
 $2$ & 30  & 70  & 160  & 30  \\ 
 $3$ & 30  & 70  & 180  & 40  \\
 $4$ & 35  & 75  & 170  & 45  \\
 $5$ & 25  & 83  & 155  & 25  \\ 
 $6$ & 30  & 88  & 175  & 25  \\ 
 $7$ & 25  & 92  & 185  & 35  \\ 
 $8$ & 30  & 88  & 145  & 35  \\ 
 \hline
\end{tabular}
\caption{The sample sizes for the different groups that were used to simulate the data.}
\label{supp-tab:ss}
\end{table}
We presented the results of clustering for small sample sizes and unbalanced sample sizes in the main document. Figure \ref{supp-fig:clustering for various sample sizes} shows the results of clustering for moderate and large sample sizes in each group.
\begin{figure}[http]
\centering
\begin{subfigure}{0.8\textwidth}
  \centering
  \includegraphics[width= \linewidth]{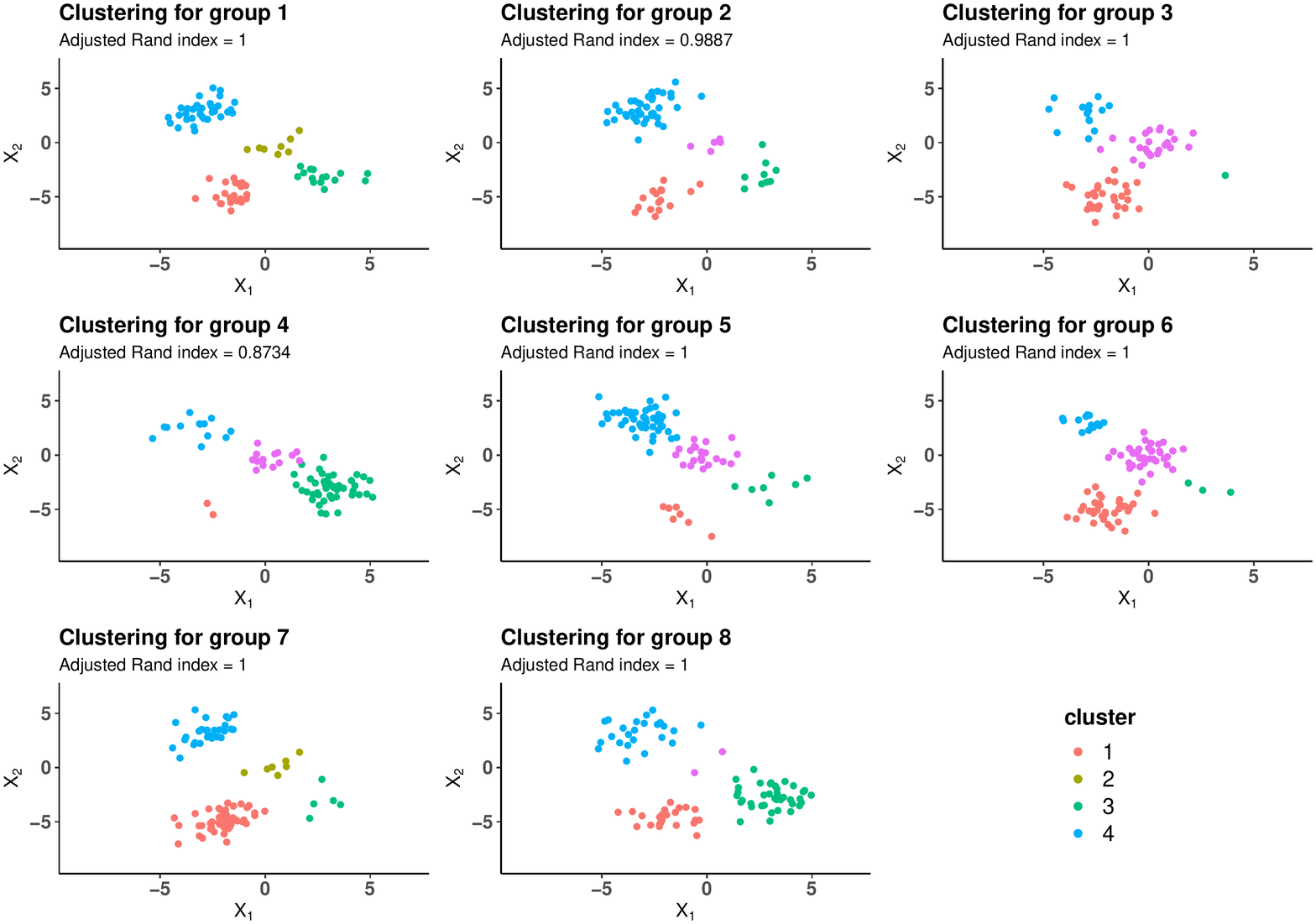}
  \caption{Moderate sample size in each group}
  \label{fig:clustering moderate}
\end{subfigure}
\par\bigskip
\begin{subfigure}{.8\textwidth}
  \centering
  \includegraphics[width=\linewidth]{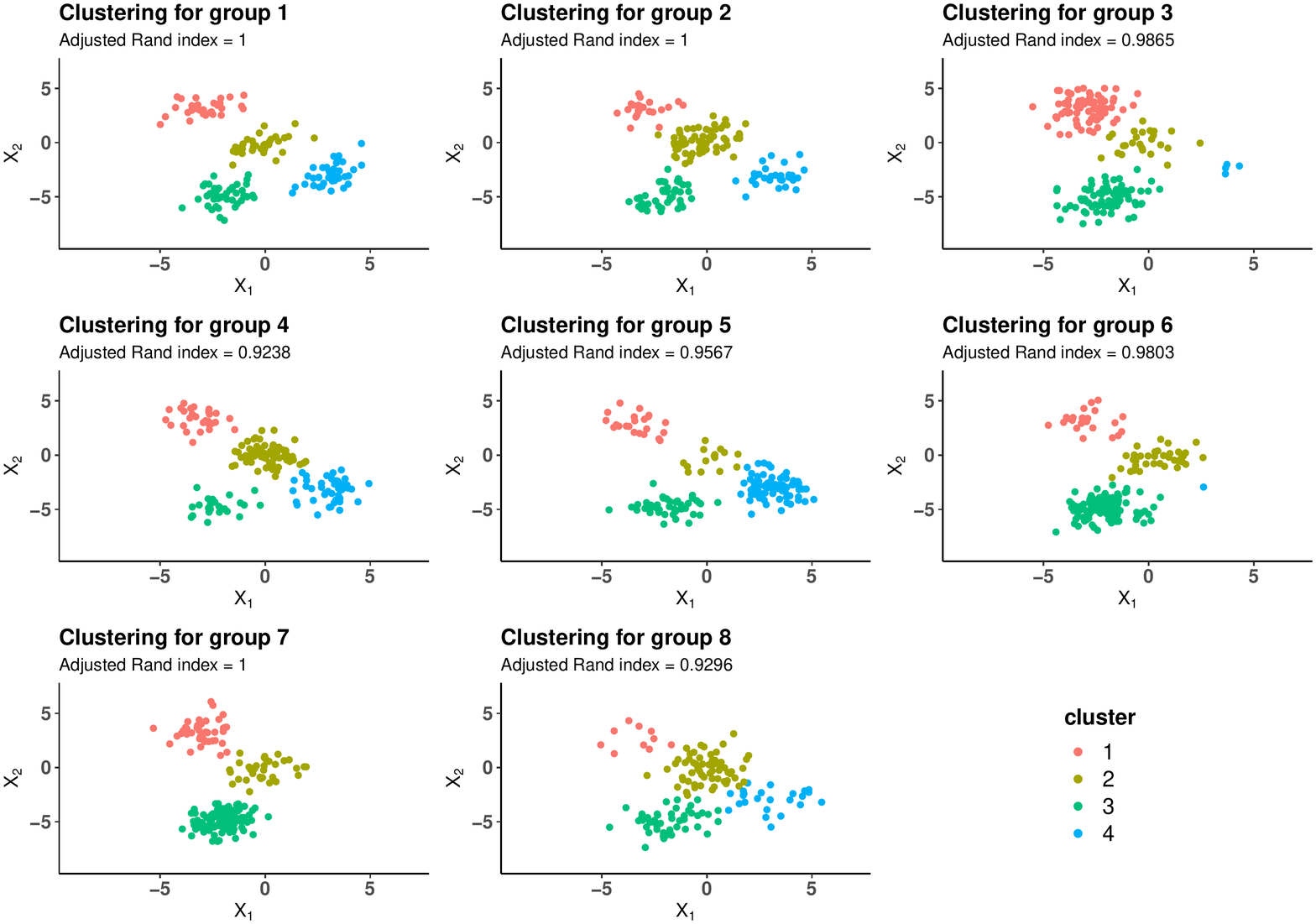}
  \caption{Large sample size in each group}
  \label{fig:clustering large}
\end{subfigure}
\caption{Clustering performance of GDP for additional sample sizes. The colors indicate the estimated clusters by GDP. Adjusted Rand index is reported at the top of each panel.}
\label{supp-fig:clustering for various sample sizes}
\end{figure}
\clearpage

\sloppy We further considered the case, wherein the simulation scenario was difficult with possibly overlapping clusters. We generated data within each of the $8$ groups from a ten-component mixture of bivariate Gaussian distributions with different covariance matrices for each group. The choice of mixture model weights for the first fours groups are summarized in Table \ref{supp-table:group weights}.

\begin{table}[htp]
    \centering
    \begin{tabular}{c|c}
    \hline
    Group & Mixture weights $\boldsymbol{\beta}_j$\\
    \hline
    1 &  $(0.100, 0.100, 0.100, 0.100, 0.100, 0.100, 0.100, 0.100, 0.100, 0.100)^\top$ \\
    2 &  $(0.167, 0.167, 0.167, 0.167, 0.167, 0.056, 0.056, 0.056, 0.000, 0.000)^\top$  \\
    3 &  $(0.095, 0.095, 0.095, 0.000, 0.000, 0.143, 0.143, 0.143, 0.143, 0.143)^\top$ \\
    4 &  $(0.030, 0.030, 0.030, 0.182, 0.182, 0.182, 0.182, 0.182, 0.000, 0.000)^\top$\\
    \hline
    \end{tabular}
    \caption{True group-specific mixture model weights.}
    \label{supp-table:group weights}
\end{table}

The mixture weights for all other groups were taken to be the mean of the mixture weights of their parent, e.g., the mixture weight for group 5 was the mean of the mixture weights of groups 2 and 4.
 The true cluster indicators $z_{ji}$'s for each of the different groups were drawn using \eqref{supp-eq::finite_mixture_model_example_appendix} and the true mixture weights. Given the cluster indicators, the data were generated from the Gaussian distribution with the true cluster-specific means $\phi_l$'s given in Table \ref{supp-table:cluster means_simulation_difficult} and the group-specific covariance matrices given in Table \ref{supp-table:population covariance_simulation}.

 \begin{table}[htp]
    \centering
    \begin{tabular}{c|c}
    \hline
    Cluster & Mean\\
    \hline
    1 &  (-2.5, 0)\\
    2 &  (0, 0)  \\
    3 &  (2.5, 0) \\
    4 &  (2.5, -2.5)\\
    5 &  (-3, -3)\\
    6 &  (2, 2)\\
    7 &  (-2, 5)\\
    8 &  (5, 8)\\
    9 &  (-5, -8)\\
    10 & (8, -8)\\
    \hline
    \end{tabular}
    \caption{True cluster-specific means.}
    \label{supp-table:cluster means_simulation_difficult}
\end{table}

In our Gibbs sampler, the truncation level of the finite mixture model was set to $L = 20$, the hyperparameter $\alpha_0$ was taken to be 1, and the base measure for GDP, $G_0$, was specified as the normal-inverse-Wishart distribution, $\mathcal{NIW}(\boldsymbol{0}, 0.01, \mathbb{I}_2, 2)$.
Upon the completion of the Gibbs sampler, the clusters were estimated by using the least squares criterion \citep{LeastSquares}, and they were compared with the true cluster labels for evaluation. We again considered a variety of sample sizes as summarized in Table \ref{supp-tab:ss}. In all cases, we ran 25,000 iterations of our Gibbs sampler and after discarding the first 15,000 samples as burn-in, considered thinning of the samples by a factor 5. The clustering results are shown in Figure \ref{supp-fig:clustering for various sample sizes difficult}. Clearly, GDP was able to identify the overlapping clusters within each group and link them across groups for all simulation scenarios with reasonable accuracy as measured by Adjusted Rand indices for each group (shown in the plots).

\begin{figure}[http]
\centering
\begin{subfigure}{0.8\textwidth}
  \centering
  \includegraphics[width= \linewidth]{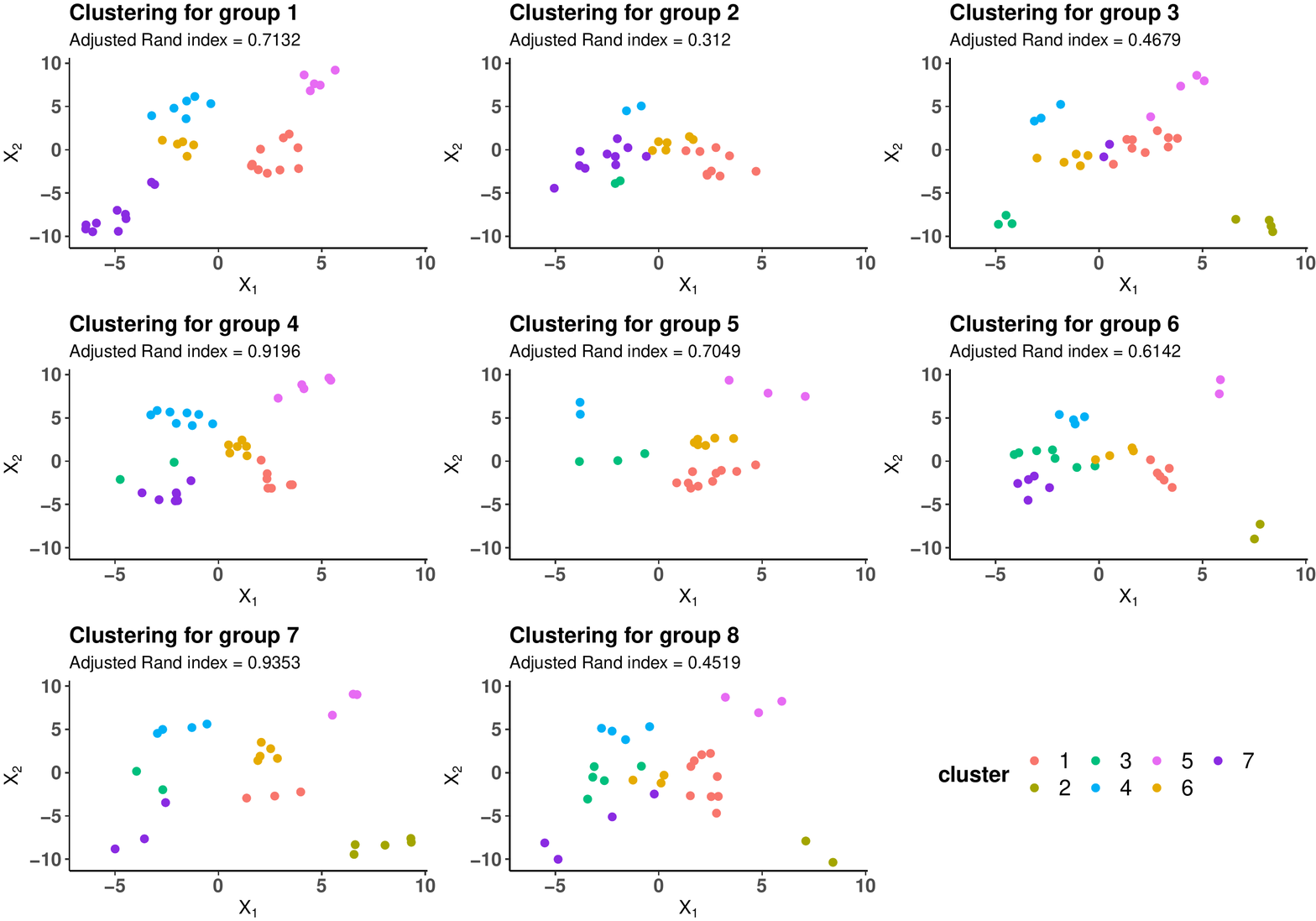}
  \caption{Small sample size in each group}
  \label{fig:clustering small difficult}
\end{subfigure}
\par\bigskip
\begin{subfigure}{.8\textwidth}
  \centering
  \includegraphics[width=\linewidth]{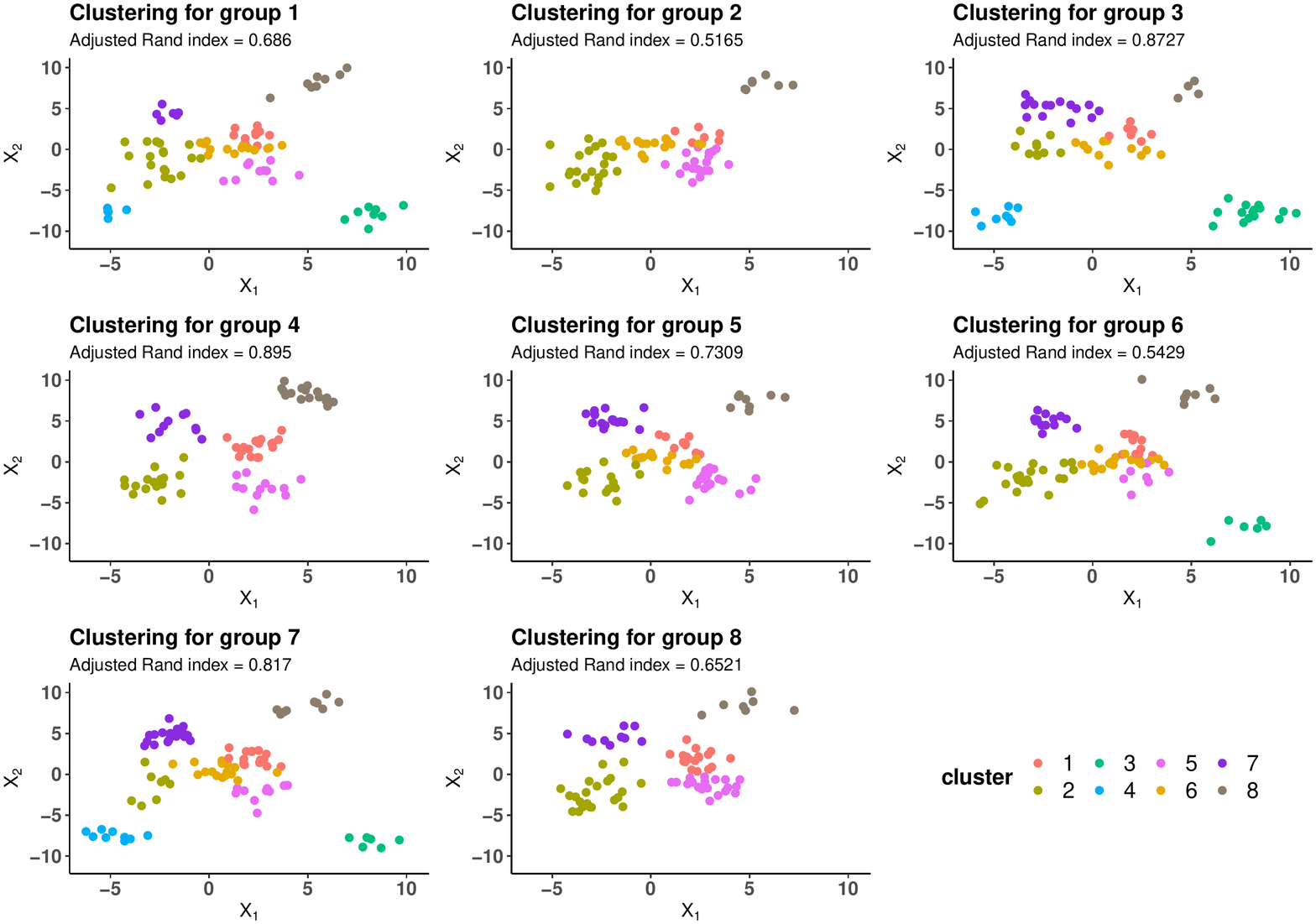}
  \caption{Moderate sample size in each group}
  \label{fig:clustering moderate difficult}
\end{subfigure}
\end{figure}
\begin{figure}[http]
\ContinuedFloat
\centering
\begin{subfigure}{0.8\textwidth}
  \centering
  \includegraphics[width= \linewidth]{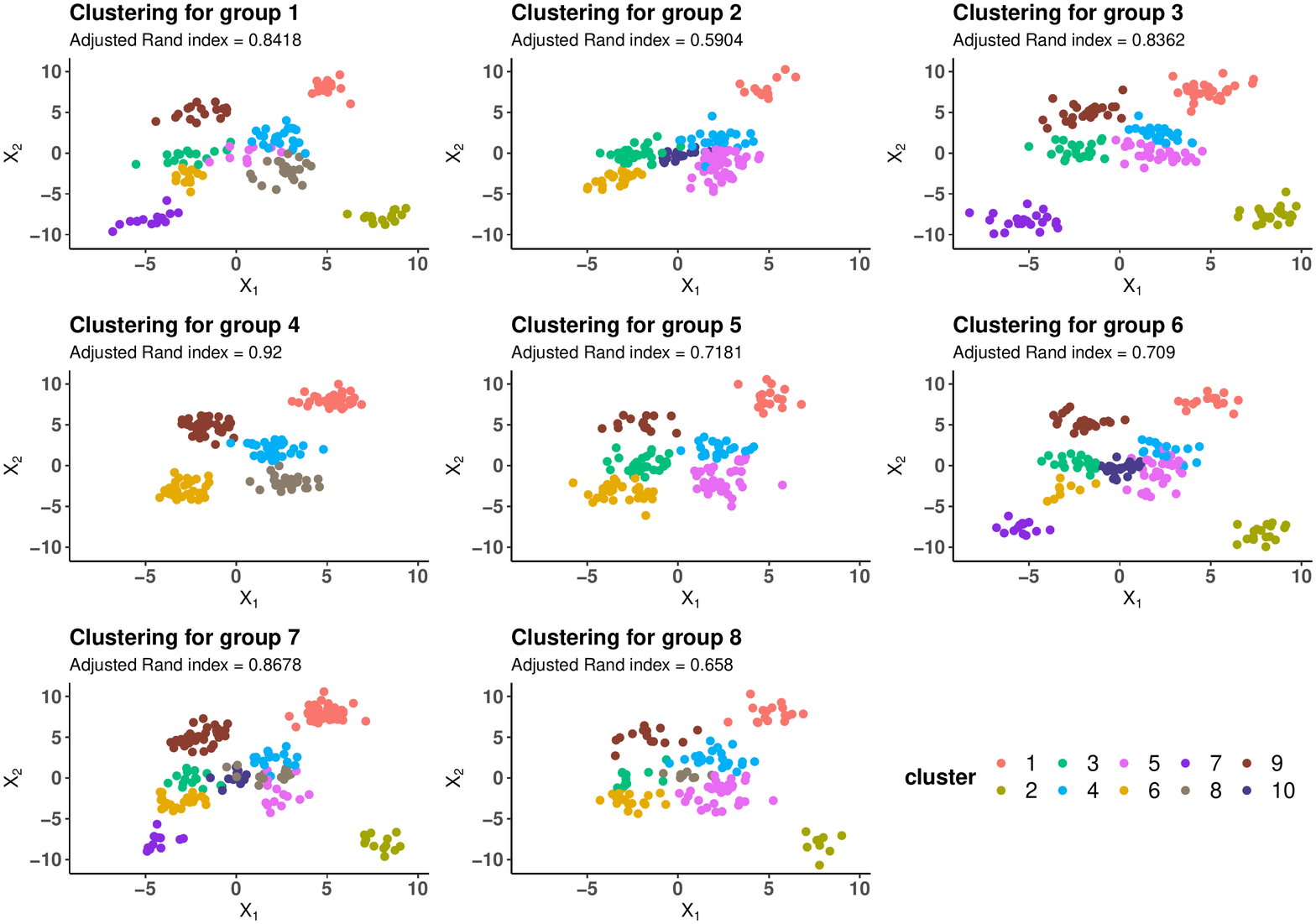}
  \caption{Large sample size in each group}
  \label{fig:clustering large difficult}
\end{subfigure}
\par\bigskip
\begin{subfigure}{.8\textwidth}
  \centering
  \includegraphics[width=\linewidth]{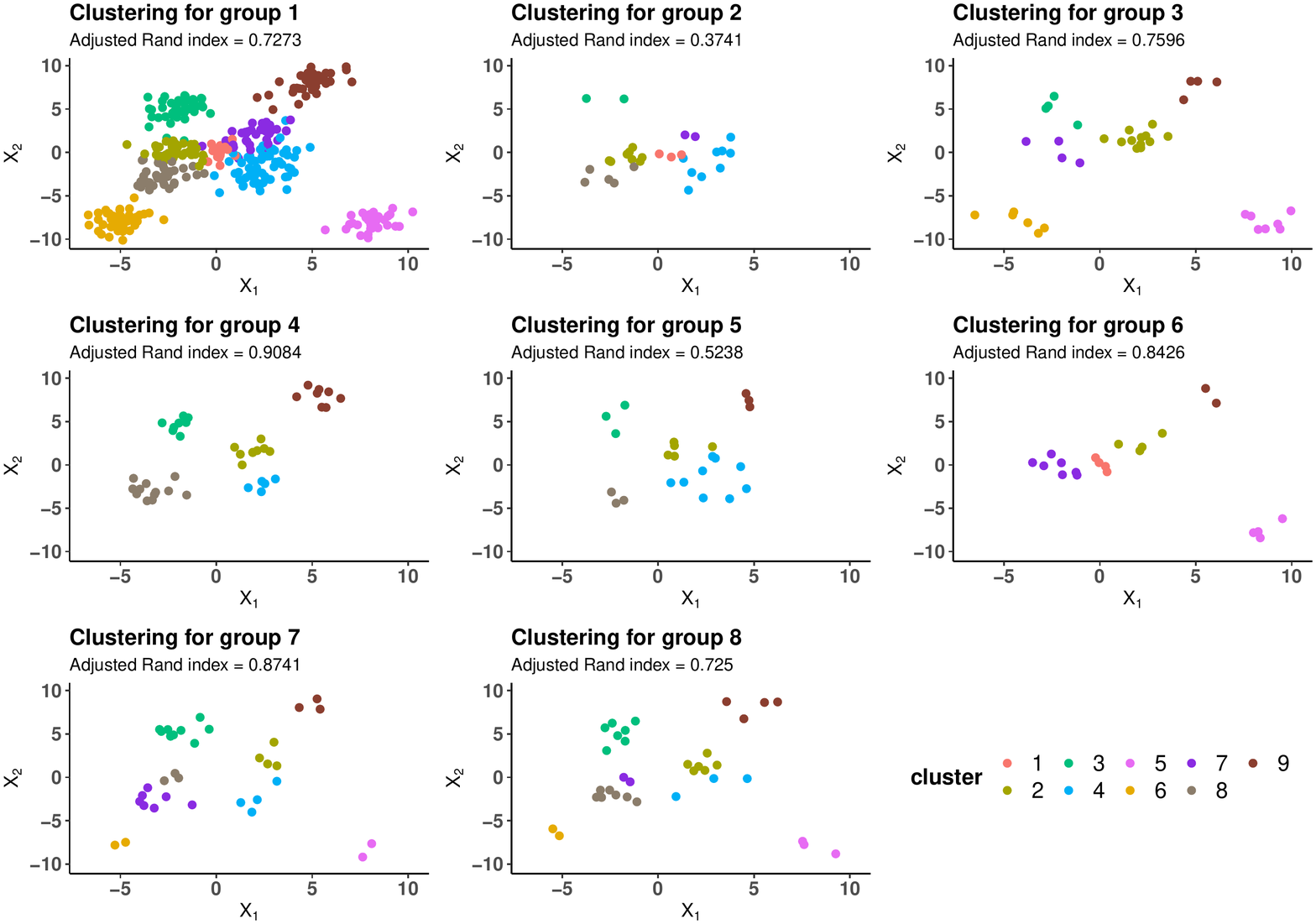}
  \caption{Unbalanced sample size in each group}
  \label{fig:clustering unbalanced difficult}
\end{subfigure}
\caption{Clustering performance of GDP for various sample sizes and difficult simulation scenario. The colors indicate the estimated clusters by GDP. Adjusted Rand index is reported at the top of each panel.}
\label{supp-fig:clustering for various sample sizes difficult}
\end{figure}

In the main manuscript, we reported the boxplot of Adjusted Rand indices for 50 replicates. Further investigation regarding the choice of $\alpha_0$ revealed no significant impact in clustering performance. 
Figure \ref{fig:gdp hdp comparison 2} shows that boxplot of Adjusted Rand indices for 50 replicates, comparing GDP, HDP, and k-means with $\alpha_0$ taken to be 6. In all situations, GDP uniformly out-performed the other two methods.

\begin{figure}[htp]
\centering
  \centering
  \includegraphics[width= 1\linewidth]{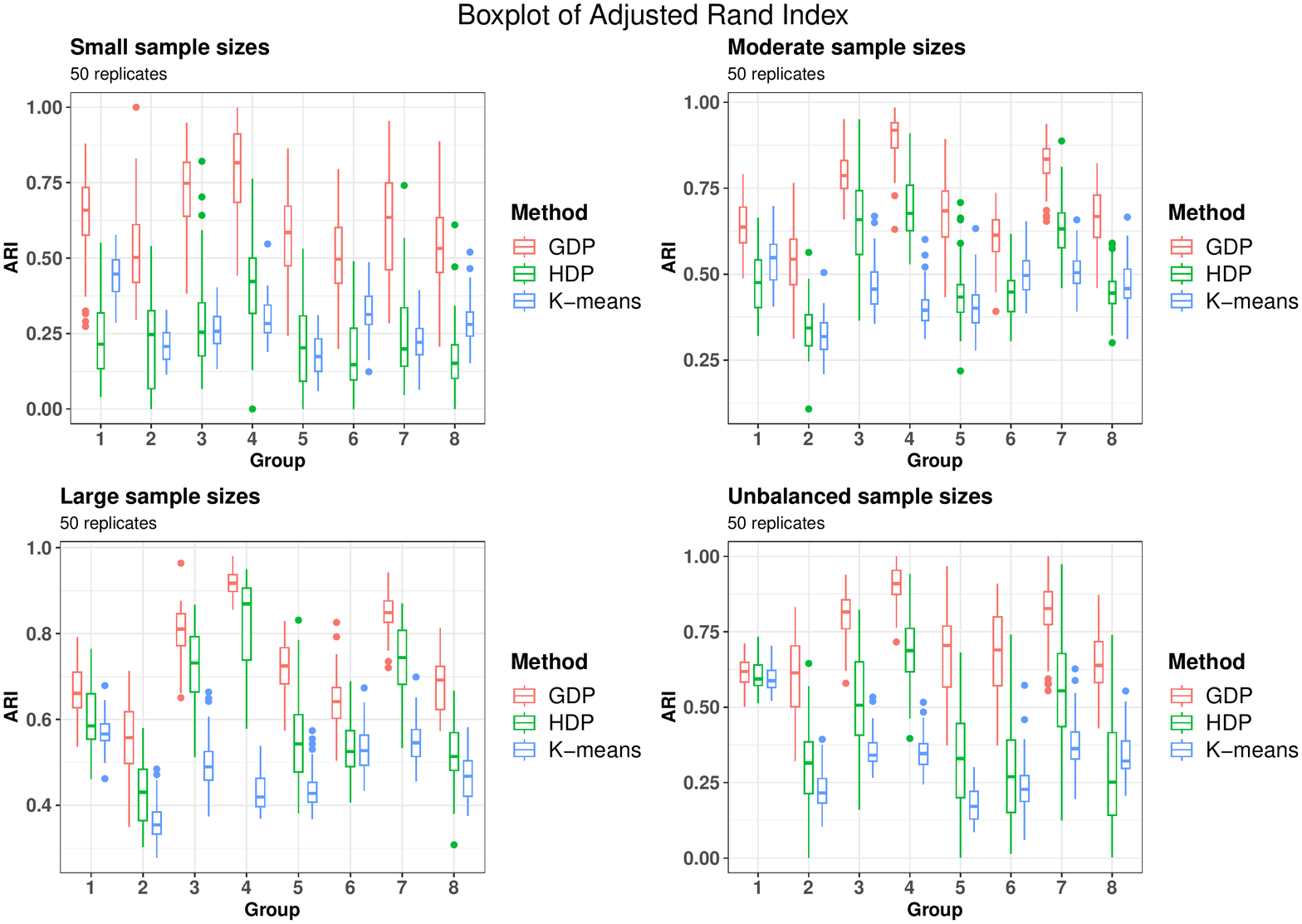}
\caption{The boxplots of the Adjusted Rand indices for GDP, HDP, and k-means for all sample sizes. In all simulations $\alpha_0$ was taken to be 6.}
\label{fig:gdp hdp comparison 2}
\end{figure}

\clearpage
\section{Real Data Analysis plots}
\label{supp-sec:real_data}
We present the traceplots (Figure \ref{supp-fig:Real data Traceplot}) of the log-likelihood for each of the four parallel chains of our sampler, corresponding to the real data analysis of the main document, after discarding the initial $25,000$ samples and thinning the samples by a factor of $15$. The traceplots indicate the presence of local modes, necessitating the need to concatenate posterior samples across these chains for more efficient and reliable inference. 

\begin{figure}[htp]
\centering
\begin{subfigure}{1\textwidth}
  \centering
  \includegraphics[width=\linewidth, scale = 1]{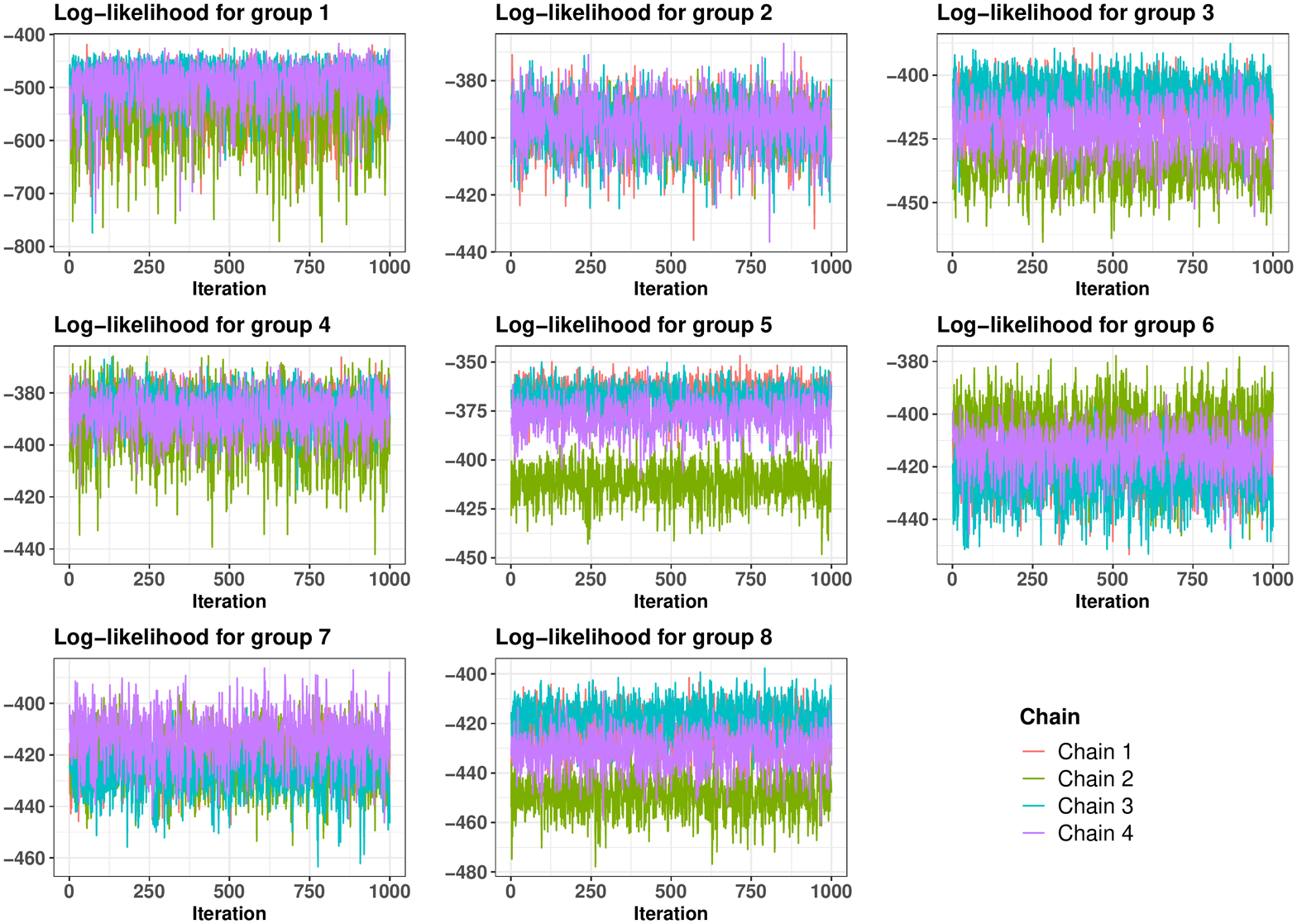}
  \caption{GDP}
  \label{supp-fig:Real data Traceplot GDP}
\end{subfigure}
\end{figure}
\begin{figure}
    \ContinuedFloat
\begin{subfigure}{1\textwidth}
  \centering
 \includegraphics[width=\linewidth, scale = 1]{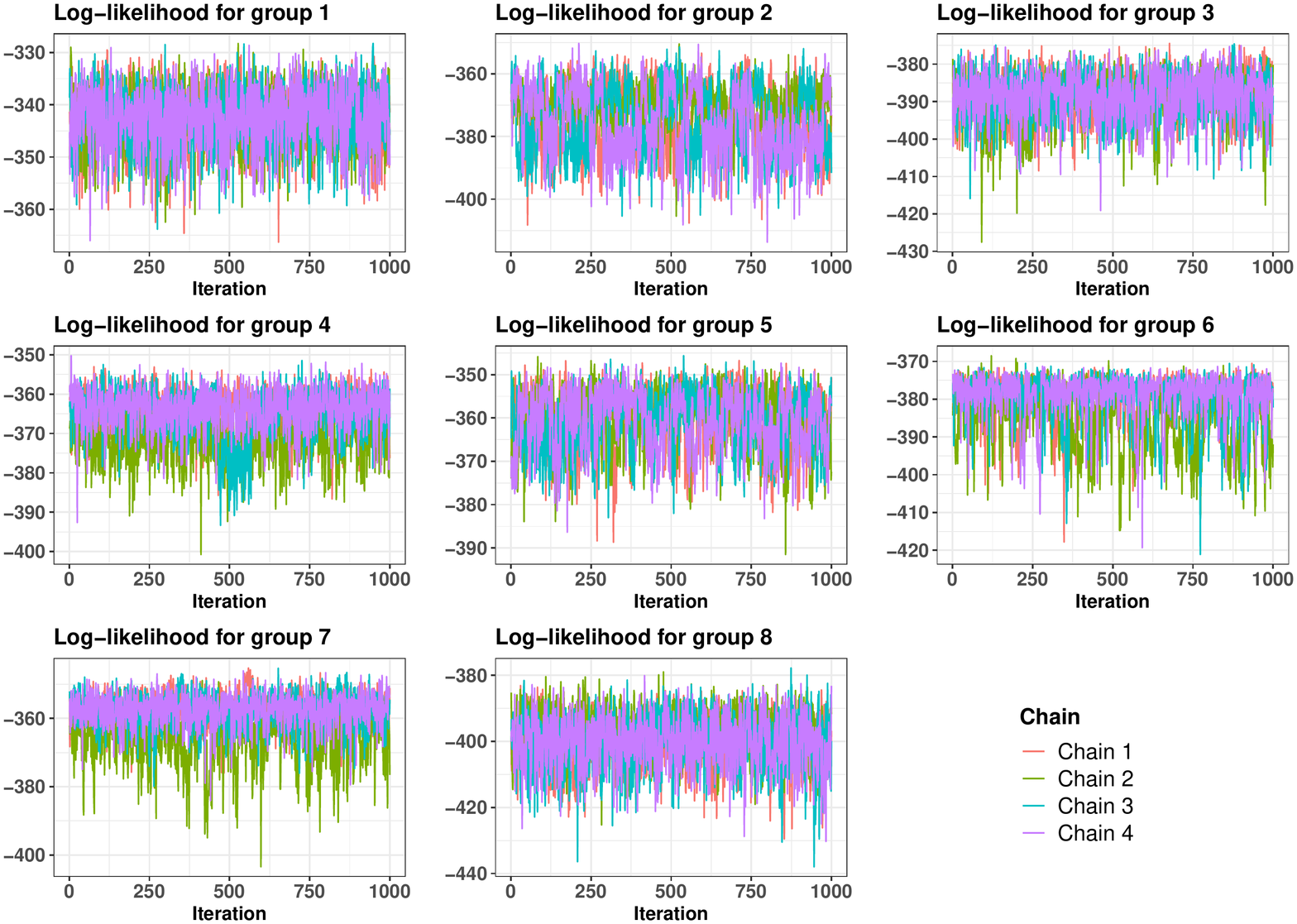}
  \caption{HDP}
  \label{supp-fig:Real data Traceplot HDP}
\end{subfigure}
\caption{Traceplots of log-likelihood for each group for (a) GDP and (b) HDP.}
\label{supp-fig:Real data Traceplot}
\end{figure}

\newpage

\bibliographystyle{myabbrvnat}
\bibliography{references}

\end{document}